\documentclass[aps,longbibliography,prx,print,twocolumn,showpacs]{revtex4-2}
\pdfoutput=1
\usepackage{amsmath}
\usepackage{amsthm}
\usepackage{amssymb}

\makeatletter
\theoremstyle{plain}
\newtheorem{thm}{\protect\theoremname}
\ifx\proof\undefined
\newenvironment{proof}[1][\protect\proofname]{\par
	\normalfont\topsep6\p@\@plus6\p@\relax
	\trivlist
	\itemindent\parindent
	\item[\hskip\labelsep\scshape #1]\ignorespaces
}{%
	\endtrivlist\@endpefalse
}
\providecommand{\proofname}{Proof}
\fi
\theoremstyle{plain}
\newtheorem{lem}[thm]{\protect\lemmaname}
\theoremstyle{definition}
\newtheorem{defn}[thm]{\protect\definitionname}

\usepackage[utf8]{inputenc}
\usepackage[english]{babel}
\usepackage[T1]{fontenc}
\usepackage{tikz}
\usepackage{lipsum}
\usepackage{xcolor}
\usepackage{nameref}

\definecolor{ForestGreen}{rgb}{0.1333,0.5451,0.1333}
\definecolor{DarkRed}{rgb}{0.8,0,0}
\definecolor{Red}{rgb}{1,0,0}

\usepackage{titletoc}

    \titlecontents{chapter}
    [0em] %
    {\medskip\bfseries}
    {\thecontentslabel}
    {}
    {\hfill\contentspage}

    \titlecontents{section}[0em]{\smallskip}%
    {\thecontentslabel\enspace}
    {}
    {\titlerule*[1.2pc]{}\contentspage}%

\usepackage{hyperref}
\hypersetup{
    bookmarksnumbered=true, 
    unicode=false, 
    pdfstartview={FitH}, 
    pdftitle={}, 
    pdfauthor={}, 
    pdfsubject={}, 
    pdfcreator={}, 
    pdfproducer={}, 
    pdfkeywords={}, 
    pdfnewwindow=true, 
    colorlinks=true, 
    linkcolor=blue, 
    citecolor=blue, 
    filecolor=blue, 
    urlcolor=blue 
}

\usepackage{cleveref}

\crefalias{thm}{theorem}

\let\oldlem\lem
\renewcommand{\lem}{
\crefalias{thm}{lemma}
\oldlem
}

\let\olddefn\defn
\renewcommand{\defn}{%
\crefalias{thm}{definition}
\olddefn
}

\AtBeginDocument{%
\let\ref\Cref
}

\makeatother

\providecommand{\definitionname}{Definition}
\providecommand{\lemmaname}{Lemma}
\providecommand{\theoremname}{Theorem}

\begin{document}
\global\long\def\id{\openone}%

\global\long\def\ket#1{\left|#1\right\rangle }%

\global\long\def\bra#1{\left\langle #1\right|}%

\global\long\def\braket#1#2{\left\langle #1|#2\right\rangle }%

\global\long\def\per#1#2{P_{#2}^{#1}}%

\global\long\def\ketbra#1#2{\ket{#1}\bra{#2}}%

\global\long\def\tr#1#2{\text{Tr}_{#1}\left[#2\right]}%

\global\long\def\norm#1{\left\Vert #1\right\Vert }%

\global\long\def\vcon{c}%

\global\long\def\vrots{n}%

\global\long\def\U#1{U_{#1}}%

\global\long\def\Uc#1{\mathcal{U}_{#1}}%

\global\long\def\eV#1{\ensuremath{#1}}%

\global\long\def\rdm{D}%

\global\long\def\match#1#2{\left|#1\cap#2\right|}%

\global\long\def\wg#1{\mathrm{Wg}^{\mathrm{\mathcal{U}}_{n}}\left(#1\right)}%

\title{Classical shadows of fermions with particle number symmetry}
\author{Guang Hao Low}
\affiliation{Microsoft Quantum, Redmond, WA 98052, USA}
\begin{abstract}
We consider classical shadows of fermion wavefunctions with $\eta$
particles occupying $n$ modes. We prove that all $k$-Reduced Density
Matrices (RDMs) may be simultaneously estimated to an average variance of
$\epsilon^{2}$ using at most $\binom{\eta}{k}\big(1-\frac{\eta-k}{n}\big)^{k}\frac{1+n}{1+n-k}/\epsilon^{2}$
measurements in random single-particle bases that conserve particle
number, and provide an estimator for any $k$-RDM with $\mathcal{O}(k^2\eta)$ classical complexity.
Our sample complexity is a super-exponential improvement over the $\mathcal{O}(\binom{n}{k}\frac{\sqrt{k}}{\epsilon^{2}})$
scaling of prior approaches as $n$ can be arbitrarily larger than
$\eta$, which is common in natural problems. Our method, in the worst-case
of half-filling, still provides a factor of $4^{k}$ advantage in
sample complexity, and also estimates all $\eta$-reduced density
matrices, applicable to estimating overlaps with all single Slater
determinants, with at most $\mathcal{O}(\frac{1}{\epsilon^{2}})$
samples, which is additionally independent of $\eta$. 
\end{abstract}
\maketitle
\tableofcontents{}

\section{Introduction}

Fermion wavefunctions are notoriously complex as their dimension scales
exponentially with particle number. Overcoming the curse of dimensionality
motivates simulating fermions as a most relevant application of quantum
computing \citep{cao19,vonBurg2020carbon}. A primary task in this
application is the characterization, or state tomography of the simulated
wavefunction. It often suffices to estimate $k$-reduced density matrices
($k$-RDMs) as arbitrary observables such as the Hamiltonian of interacting
electrons, dipole moment, or other multi-point correlation functions
\citep{Huggins2022FermionicQMC} are supported by this basis.

State tomography protocols based on measuring quantum states in random
bases \citep{Jeongwan2017tomography,Elben2020topologicalinvariants}
have recently advanced dramatically. The celebrated concept of classical
shadows \citep{Huang2020shadowestimation,Chen2021robustshadow} tailors
these random bases to estimate a targeted family of observables to
some variance $\epsilon^{2}$ with optimal sample complexity. For
instance, random single-qubit bases on $n$ qubits simultaneously
estimate \emph{all} $k$-local Pauli operators from the same set of
$\mathcal{O}(3^{k}/\epsilon^{2})$ samples. Even though there are
$3^{k}\binom{n}{k}$ such Paulis, the cost of estimation incredibly
depends only on the size of their support and is independent of $n$.

By representing fermions with qubits through, say, the Jordan-Wigner
or Bravyi-Kitaev encoding \citep{Bravyi2002Fermion}, state tomography
on fermions is reduced to one of the many optimal schemes for qubit
tomography. The most compact encoding \citep{Jiang2020fermionmapternery}
then estimates all $k$-RDMs on $n$ fermion modes using only $\mathcal{O}((2n)^{k}/\epsilon^{2})$
samples. As bootstrapping to a qubit-based scheme incurs substantial
overhead, directly randomizing the algebra of fermions, such as with
fermionic gaussians, further reduces the samples required to only
$\binom{n}{k}\sqrt{\pi k}/\epsilon^{2}$ \citep{Zhao2020fermionshadows}.
As there are $\mathcal{O}(n^{2k})$ independent $k$-RDMs and $\mathcal{O}(n^{k})$
mutually commuting observables, this result is optimal, but does not
realize the super-exponential improvement seen in the qubit setting.

Discovering an analogous super-exponential improvement would unlock
for fermions many applications found in the toolbox of randomized
qubit measurements \citep{Elben2022RandomziedToolbox}. A missing
ingredient is the particle number symmetry present in many systems
of interest, ranging from electronic structure to the Hubbard model.
The design of random bases targeting number-conserving $k$-RDMs should
account for this crucial prior that the fermion wavefunction has a
definite number of $\eta$ particles occupying these $n$ modes. As
the condition that $n\gg\eta$ also occurs naturally, such as in modeling
dynamical correlation \citep{Reiher2016Reaction} or in plane-wave
simulations \citep{Babbush2018encoding,Low2018IntPicSim}, a scheme
that scales with $\eta$ instead of $n$ is highly desirable. However,
prior approaches target all $k$-RDMs, both number-conserving and
not. One solution is to choose random bases that are also number-conserving,
but this appears challenging. For instance, prior attempts \citep{Zhao2020fermionshadows} required bootstrapping to qubit protocols
to achieve tomographic completeness and ultimately achieved the same
$\mathcal{O}(n{}^{k}/\epsilon^{2})$ sample complexity, or was able to prove scaling with $\eta$ for only for the $k=1$ case~\cite{Naldesi2023Fermionic}.

In this work, we successfully exploit particle number symmetry. We
find random bases that simultaneously estimate all $\binom{n}{k}^{2}$
independent number-conserving $k$-RDMs 
\begin{equation}
\langle\rdm_{\vec{q}}^{\vec{p}}\rangle=\tr{}{\rdm_{\vec{q}}^{\vec{p}}\rho},\quad\rdm_{\vec{q}}^{\vec{p}}\doteq a_{p_{1}}^{\dagger}\cdots a_{p_{k}}^{\dagger}a_{q_{k}}\cdots a_{q_{1}},
\end{equation}
of any quantum state $\rho$ with an average variance of $\epsilon^{2}$
using only $\mathcal{O}(\eta^{k}/\epsilon^{2})$ samples, where the
fermion operators satisfy the usual anti-commutation relations $\{a_{j},a_{k}^{\dagger}\}=\delta_{jk}$,
$\{a_{j},a_{k}\}=\{a_{j}^{\dagger},a_{k}^{\dagger}\}=0$. As the number
of fermion modes can be arbitrarily larger than the number of particles,
this is a super-exponential improvement in complete analogy to the
qubit setting. 

\begin{thm}
\label{thm:main_theorem}For any $\eta$-particle $n$-mode fermion
state $\rho$, let the unitary single-particle basis rotation be 
\begin{equation}
U_{\eta}(u)\doteq e^{\sum_{p,q=1}^{n}\left[\ln(u)\right]_{pq}a_{p}^{\dagger}a_{q}}\in\mathbb{C}^{\binom{n}{\eta}\times\binom{n}{\eta}},\label{eq:fermion_randomization}
\end{equation}
where $u\in\mathbb{C}^{n\times n}$ is a Haar random unitary. Measure
$U_{\eta}(u)\rho U_{\eta}^{\dagger}(u)$ to obtain the occupation
$\vec{z}\in\mathcal{S}_{n,\eta}\doteq\left\{ (z_{1},\cdots,z_{\eta}):\forall j\in[\eta],1\le z_{j}<z_{j+1}\le n\right\} $
with probability $\bra{\vec{z}}U_{\eta}(u)\rho U_{\eta}^{\dagger}(u)\ket{\vec{z}}$
and let $v_{\vec{z}}\in\mathbb{Z}^{n\times n}$ be any permutation
matrix that maps elements of $[\eta]$ to $\vec{z}$. Then the single-shot
estimator of any $k$-RDM is
\begin{align}
\langle\hat{\rdm}_{\vec{q}}^{\vec{p}}\rangle & =\bra{\vec{q}}U_{k}^{\dagger}(v_{\vec{z}}^{\dagger}u)E_{\eta,k}U_{k}(v_{\vec{z}}^{\dagger}u)\ket{\vec{p}},\label{eq:k-RDM_estimator}\\
E_{\eta,k} & \doteq\sum_{\vec{r}\in\mathcal{S}_{n,k}}\ketbra{\vec{r}}{\vec{r}}\frac{\binom{\eta-s'}{k-s'}\binom{n-\eta+s'}{s'}}{(-1)^{k+s'}\binom{k}{s'}},
\end{align}
where $s'=\left|\vec{r}\cap[\eta]\right|$ is the number of elements
$\vec{r}$ and $(1,\cdots,\eta)$ share in common, with average variance
\begin{align}
\mathcal{V}=\mathbb{E}_{\vec{p},\vec{q}}\big[\mathrm{Var}[\langle\hat{\rdm}_{\vec{q}}^{\vec{p}}\rangle]\big]\le\frac{\tr{}{E_{\eta,k}^{2}}}{\binom{n}{k}^{2}}-\frac{\binom{n-k}{\eta-k}^{2}}{\binom{n}{\eta}^{2}\binom{n}{k}}.\label{eq:estimate_one_observable}
\end{align}
\end{thm}
Importantly, the symmetries of our estimator allow it to be evaluated
efficiently in $\mathcal{O}(k^{2}\eta)$ time~\footnote{An estimator with $\mathcal{O}(N^{5})$ classical runtime by polynomial interpolation was proposed in the original preprint of this manuscript.}, also independent of
$n$, even through the naive approach of multiplying dimension $\binom{n}{k}\times\binom{n}{k}$
matrices is efficient only for constant $k$. This is through a reduction
to evaluating Pfaffians corresponding to traces of products of fermionic
gaussians \citep{Terhal2002fermions,Bravyi2005}, which is similar
to recent work \citep{Wan2022MatchgateShadows} but significantly
faster due to hidden structure in our number-conserving case.

We obtain our claim on sample complexity by averaging over $N$ independent
sampled pairs $\left(u,\vec{z}\right)$. An upper bound on the variance
in \ref{eq:estimate_one_observable} is
\begin{align}
\mathcal{V} & \le\binom{\eta}{k}\left(1-\frac{\eta-k}{n}\right)^{k}\frac{1+n}{1+n-k}.\label{eq:variance_upper_bound}
\end{align}
In the worst-case of half-filling $\eta=n/2$, with large $n$ and
fixed $k$, this bound also implies a sample complexity $\mathcal{V}=\frac{1}{2^{k}}\binom{n/2}{k}\left(1+\mathcal{O}\left(\frac{k^{2}}{n}\right)\right)$
which is an exponential factor of $4^{k}$ smaller than the prior
approaches in the common case $k^{2}\ll n$. Even for small $k$,
this reduction is highly relevant to practical implementations of
quantum algorithms such as the variational quantum eigensolver \citep{McClean2016eigensolver}.
The case of very large $k=\eta$ is also of interest to applications such as quantum-classical auxiliary-field quantum Monte Carlo \citep{Huggins2022FermionicQMC}, and we prove in that case $\mathcal{V}\le\frac{4}{3\epsilon^2}$, compared to prior art of $\tilde{\mathcal{O}}(\sqrt{n}/\epsilon^2)$ \citep{Wan2022MatchgateShadows}.
Moreover, our scheme is practical as the quantum circuits
implementing each basis rotation can have depth as little as $\mathcal{O}(n)$
\citep{Kivlichan2018GivensLinearDepth}. Notably, measurements in
random single-particle bases simultaneously reveals information on
both local and non-local observables where $k$ and $\eta-k$ is small
respectively.

The simple and exact expression for our estimator belies a complicated
analysis. Prior qubit analyses are greatly simplified by how the easily
implementable group of random Clifford bases are a unitary $3$-design
\citep{Webb2016Clifford3Design,Zhu2017Clifford3Design} on the entire
state space. Unfortunately for fermions, even though random single-particle
rotations are generated by Haar random unitaries, they fail to be
$t$-designs on the entire state space except when $\eta=1$ \citep{Wan2022MatchgateShadows}
or $k=1$ \citep{Naldesi2022FermionicCorrelation}. We prove \ref{thm:main_theorem}
in four key steps.

In \ref{sec:rdm_basis}, the basic observation in the classical shadows
framework is that averaging over all classical shadows in the random
basis $U\in\mathcal{U}$ defines a measurement channel
\begin{align}
\mathcal{M}_{\mathcal{U}}\left[\rho\right] & \doteq\mathbb{E}_{\vec{z},U}\left[U^{\dagger}\ketbra{\vec{z}}{\vec{z}}U\right]\nonumber \\
 & =\tr{}I\tr 1{\mathcal{T}_{2,\mathcal{U}}\left(\rho\otimes I\right)},\label{eq:measurement_channel}
\end{align}
expressed in terms of a $t$-fold twirling channel acting on the basis
state $\ketbra{\vec{z}}{\vec{z}}$ like
\begin{equation}
\mathcal{T}_{t,\mathcal{U}}\doteq\int\left(U\ketbra{\vec{z}}{\vec{z}}U^{\dagger}\right)^{\otimes t}\mathrm{d}U_{\mathrm{Haar}}(\mathcal{U}).\label{eq:twirling_channel}
\end{equation}
In our case $U_{\eta}(u)\in\mathcal{U}=\wedge^{\eta}\mathcal{U}_{n}$
is the group of single-particle rotations, where $u\in\mathcal{U}_{n}$
is the $n$-dimension unitary group. So long as $\mathcal{U}$ is
tomographically complete for $\rho$, the measurement channel may
be inverted on any classical shadow to form a single-shot unbiased
estimate of $\rho=\mathbb{E}_{u,\vec{z}}\left[\hat{\rho}_{u,\vec{z}}\right],$
where
\begin{align}
\hat{\rho}_{u,\vec{z}} & \doteq\mathcal{M}_{\mathcal{U}}^{-1}\left[U_{\eta}^{\dagger}(u)\ketbra{\vec{z}}{\vec{z}}U_{\eta}(u)\right].\label{eq:state_estimator}
\end{align}
Hence, we demonstrate in \ref{thm:tomographic_completeness} that
$\wedge^{\eta}\mathcal{U}_{n}$ is tomographically complete, contrary
to a previous negative result where the $u$ are restricted to permutations
\citep{Zhao2020fermionshadows}. In fact, it suffices to just invert
the measurement channel on the state $\ketbra{[\eta]}{[\eta]}$, as
\begin{equation}
\hat{\rho}_{u,\vec{z}}=U_{\eta}^{\dagger}(v_{\vec{z}}^{\dagger}u)\mathcal{M}_{\mathcal{U}}^{-1}\left[\ketbra{[\eta]}{[\eta]}\right]U_{\eta}(v_{\vec{z}}^{\dagger}u),
\end{equation}
following the existence of a permutation $v_{\vec{z}}$ such that
$\ket{\vec{z}}=U_{k}(v_{\vec{z}})\ket{[\eta]}$, and the invariance
of Haar integration with respect to a change of variables.

In \ref{sec:fermion_shadows}, we find a closed-form expression for
the twirling channel $\mathcal{T}_{2,\wedge^{\eta}\mathcal{U}_{n}}$
for all $n$ and $\eta$. We leave the $t=3$ case to future work,
which would enable a per-$k$-RDM variance analysis rather than an
average. This allows us to identify in \ref{thm:inverse_measurement_channel}
the eigenoperators $\tilde{n}_{\vec{x},\vec{y}}$ of the measurement
channel
\begin{align}
\mathcal{M}_{\mathcal{U}}\left[\tilde{n}_{\vec{x},\vec{y}}\right] & =\binom{n+1}{d}^{-1}\tilde{n}_{\vec{x},\vec{y}},\label{eq:measurement_channel_eigenvalue}\\
\tilde{n}_{\vec{x},\vec{y}} & \doteq\prod_{j=1}^{d}\left(\hat{n}_{x_{j}}-\hat{n}_{y_{j}}\right),\label{eq:eigenoperator}
\end{align}
where $\hat{n}_{j}=a_{j}^{\dagger}a_{j}$ are number operators and
$\vec{x}\cap\vec{y}=\emptyset$. By expressing $\ketbra{[\eta]}{[\eta]}=\sum_{\vec{x},\vec{y}}c_{\vec{x},\vec{y}}\tilde{n}_{\vec{x},\vec{y}}$
as a linear combination of $\tilde{n}_{\vec{x},\vec{y}}$, we successfully
find the inverse $\mathcal{M}_{\mathcal{U}}^{-1}\left[\ketbra{[\eta]}{[\eta]}\right]=\binom{n+1}{d}\sum_{\vec{x},\vec{y}}c_{\vec{x},\vec{y}}\tilde{n}_{\vec{x},\vec{y}}$. 

In \ref{sec:efficient_estimation_from_shadows}, as the estimate $\hat{\rho}_{u,\vec{z}}$
has exponentially large dimension, finding an expression for it does
not guarantee the efficient computation of arbitrary observables $O$.
Fortunately, efficient computation is guaranteed if the estimate $\langle\hat{O}\rangle=\binom{n+1}{d}\sum_{\vec{x},\vec{y}}c_{\vec{x},\vec{y}}\tr{}{\tilde{n}_{\vec{x},\vec{y}}O}$
simplifies into an implicit sum over polynomially many terms without
explicitly forming $\hat{\rho}_{u,\vec{z}}$. For any observable that
is a linear combination of $k$-RDMs
\begin{align}
O & =\sum_{\vec{p},\vec{q}\in\mathcal{S}_{n,k}}o_{\vec{p}.\vec{q}}\rdm_{\vec{q}}^{\vec{p}},\label{eq:observable_linear_combination_rdm}
\end{align}
we show in \ref{thm:single_shot_estimate_efficient} that the estimator
for $\hat{O}=\tr{}{oU_{k}^{\dagger}(v_{\vec{z}}^{\dagger}u)E_{\eta,k}U_{k}(v_{\vec{z}}^{\dagger}u)}$,
with \ref{eq:k-RDM_estimator} as a special case when $O$ is a single
$k$-RDM. This expression highlights how our estimator has no preferred
basis. Though we specify $k$-RDMs in the computational basis, any
basis rotated $k$-RDMs, e.g. $U_{k}(w)\hat{\rdm}_{\vec{q}}^{\vec{p}}U_{k}^{\dagger}(w)$,
which could contain exponentially many terms, may be estimated just
as easily by absorbing $v_{\vec{z}}^{\dagger}u\rightarrow v_{\vec{z}}^{\dagger}uw$.
One might also think that this estimator is computationally efficient
only for constant $k$ as $U_{k}(u)$ has dimension $\binom{n}{k}^{2}$.
However, we show in \ref{sec:Pffaffians} that any $k$-RDM may be
evaluated in just $\mathcal{O}(k^{2}\eta)$ time, which is also independent
of $n$.

Finally, in \ref{sec:shadow_norm}, the classical shadows framework
states that the variance of any observable for a given quantum state
$\mathrm{Var}\left[\hat{O}\right]\le\norm O_{\text{s},\rho}^{2}$
is bounded above by the square of a quantity called the shadow norm.
\begin{align}
\norm O_{\text{s,\ensuremath{\rho}}}^{2} & \doteq\mathbb{E}_{u,\vec{z}}\left[|\langle\hat{O}_{\mathrm{tr}}\rangle|^{2}\right],\label{eq:shadow_norm}
\end{align}
where $O_{\mathrm{tr}}$ is the trace-free component of $O$. The
convention is to express $\norm O_{\text{s},\rho}^{2}=\tr{}{\mathcal{T}_{3,\wedge^{\eta}\mathcal{U}_{n}}\left(\rho\otimes\mathcal{M}_{\mathcal{U}}^{-1}\left(O_{\mathrm{tr}}\right)\otimes\mathcal{M}_{\mathcal{U}}^{-1}\left(O_{\mathrm{tr}}\right)\right)}$
in terms of the $3$-fold twirling channel, which appears quite difficult
to evaluate. We achieve some simplification by considering the case
where $O$ is a single $k$-RDM. Substituting our estimator \ref{eq:k-RDM_estimator}
into \ref{eq:shadow_norm}, we obtain in \ref{thm:error_bound} a
state-independent equality on the squared shadow norm averaged over
all $k$-RDMs
\begin{equation}
\mathbb{E}_{\vec{p},\vec{q}}[\|\rdm_{\vec{q}}^{\vec{p}}\|_{\text{s},\rho}^{2}]=\frac{\tr{}{E_{\eta,k}^{2}}}{\binom{n}{k}^{2}}-\frac{\binom{n-k}{\eta-k}^{2}}{\binom{n}{\eta}^{2}\binom{n}{k}},\label{eq:shadow_norm_2}
\end{equation}
seen in \ref{eq:estimate_one_observable} without needing to evaluate
$\mathcal{T}_{3,\wedge^{\eta}\mathcal{U}_{n}}$. The upper bound \ref{eq:variance_upper_bound}
on this exact expression for the squared average shadow norm is one
that works reasonably well when $\eta\ll n$. Given the variance,
we note that taking the median-of-means \citep{Huang2020shadowestimation}
or mean of multiple estimates \citep{Zhao2020fermionshadows} furnishes
the additive error bounds commonly seen in related work.

In \ref{sec:Conclusion}, we discuss other implications of our results,
such as estimating the overlap with all Slater determinants with $\frac{16}{3\epsilon^{2}}$
samples on average, and highlight possible future directions. We relegate
some of our longer and highly technical proofs to the appendices.
\ref{sec:Haar_integrals} evaluates $\mathcal{T}_{2,\wedge^{\eta}\mathcal{U}_{n}}$
in closed-form by a combinatorial approach, and \ref{sec:Hypergeometric-sums}
evaluates certain triple- and quadruple-sums over hypergeometric terms
related to finding the coefficients $c_{\vec{x},\vec{y}}$ and entries
of the estimation matrix. We define some commonly used notation in
the remainder of this section. 

\subsection{Notation}

Let $a_{k}^{\dagger}\ket 0\doteq\ket k$ be basis elements of a dimension
$n$ complex vector space $V_{n}$. The basis elements of the $\eta$
fermion space $\wedge^{\eta}V_{n}$ are $\ket{\eV{\vec{z}}}\doteq\bigwedge_{k=1}^{\eta}\ket{\eV{z_{k}}}=\ket{\eV{z_{1}}}\wedge\cdots\wedge\ket{\eV{z_{\eta}}}$
with dimension $\mathrm{dim}\left[\wedge^{\eta}V_{n}\right]=\binom{n}{\eta}$.
These are indexed by the occupation number basis $\eV{\vec{z}}\in\mathcal{S}_{n,\eta}$,
where
\begin{align}
\mathcal{S}_{n,\eta} & \doteq\left\{ (z_{1},\cdots,z_{\eta}):\forall j\in[\eta],1\le z_{j}<z_{j+1}\le n\right\} ,
\end{align}
is the set of $\eta$ sorted integers between $1$ and $n$. For example,
$[\eta]=(1,\cdots,\eta)\in\mathcal{S}_{n,\eta}$, and using the set
complement notation, $[n]\backslash[\eta]=(\eta+1,\cdots,n)\in\mathcal{S}_{n,\eta}$.
The basis state $\ket{\vec{z}}$ is an eigenvector of the number operators
$\hat{n}_{j}=a_{j}^{\dagger}a_{j}$, where products of $d$ distinct
number operators $\hat{n}_{\vec{p}}=\hat{n}_{p_{1}}\dots\hat{n}_{p_{d}}$
satisfy
\begin{equation}
\bra{\vec{z}}\hat{n}_{\vec{p}}\ket{\vec{z}}=\begin{cases}
1, & \vec{p}\subseteq\vec{z},\\
0, & \text{otherwise},
\end{cases}
\end{equation}
and are rank-$\binom{n-d}{\eta-d}$ projectors. For instance, for
all $\vec{z}\in\mathcal{S}_{n,\eta}$, the rank-$1$ projector
\begin{equation}
\ketbra{\eV{\vec{z}}}{\eV{\vec{z}}}=\Pi_{\vec{z}}=\hat{n}_{\vec{z}}.
\end{equation}
We will often perform sums over combinations of $\hat{n}_{\vec{z}}$
such as
\begin{equation}
e_{k}\left(\hat{n}_{1},\cdots,\hat{n}_{d}\right)\doteq\sum_{\vec{z}\in\mathcal{S}_{d,k}}\hat{n}_{\vec{z}}=\frac{1}{k!}\sum_{\vec{z}\in S_{[d],k}}\hat{n}_{\vec{z}},
\end{equation}
which define elementary symmetric polynomials of degree $k$ in $d$
commuting variables, where we introduce $S_{\vec{p},k}$ as all permutations
of $k$ elements of $\vec{p}$. To simplify notation, we may replace
the argument $e_{k}\left(\{\hat{n}\}_{[d]}\right)\doteq e_{k}\left(\hat{n}_{1},\cdots,\hat{n}_{d}\right)$.

Let $u\in\mathcal{U}_{n}$ be a unitary operation on basis elements
of the complex vector space $\mathbb{C}^{n}$, that is $\ket{u_{k}}\doteq u\ket k=\sum_{q}u_{qk}\ket{\eV q}.$
The representation $U_{\eta}(u)$ on $\wedge^{\eta}\mathcal{U}_{n}$
is then a unitary fermion rotation applying 
\begin{align}
U_{\eta}(u)\ket{\eV{\vec{z}}} & =\bigwedge_{k=1}^{\eta}\ket{u_{z_{k}}}=\sum_{\vec{p}\in\mathcal{S}_{n,\eta}}\det\left[u_{\vec{p}\vec{z}}\right]\ket{\eV{\vec{p}}},\label{eq:fermion_rotation}
\end{align}
where $u_{\vec{x}\vec{y}}$ denotes the submatrix formed by taking
rows $x_{1},x_{2},\cdots$ and columns $y_{1},y_{2},\cdots$ of $u$.
In other words, the determinant $\det\left[u_{\vec{p}\vec{z}}\right]$
is the determinant of a matrix minor of $u$. Fermion rotations are
a homomorphism of $\mathcal{U}_{n}$ as 

\begin{align}
U_{\eta}(v)U_{\eta}(u) & =U_{\eta}(vu).\label{eq:homomorphism}
\end{align}
Fermion rotations are also known as single-particle basis rotations
as each creation operator is rotated to a linear combination of other
creation operators as follows. 
\begin{align}
U_{\eta}(u)a_{k}^{\dagger}U_{\eta}^{\dagger}(u)= & \sum_{j}u_{jk}a_{j}^{\dagger},\\
U_{\eta}^{\dagger}(u)a_{k}^{\dagger}U_{\eta}(u)= & \sum_{j}u_{kj}^{*}a_{j}^{\dagger}.
\end{align}
Consider the special case where $v_{\vec{z}}\in\mathbb{Z}^{n\times n}$
is any permutation matrix such that $\ket{z_{j}}=v_{\vec{z}}\ket j$
for all $j\in[\eta]$. Then 
\begin{align}
\ket{\vec{z}} & =U_{\eta}(v_{\vec{z}})\ket{[\eta]}.
\end{align}
We often rotate product of creation and annihilation operators. For
any $\vec{p},\vec{q}\in\mathcal{S}_{n,k}$, the rotated $k$-RDM is
\begin{align}
U_{\eta}(u)D_{\vec{q}}^{\vec{p}}U_{\eta}^{\dagger}(u) & =\sum_{\vec{p}',\vec{q}'\in\mathcal{S}_{n,k}}\det[u_{\vec{p}'\vec{p}}]D_{\vec{q}'}^{\vec{p}'}\det[(u^{\dagger})_{\vec{q}\vec{q}'}].\label{eq:rdm_rotation}
\end{align}
Given a linear combination $k$-RDMs \ref{eq:observable_linear_combination_rdm},
the rotated observable $O'=U_{\eta}(u)OU_{\eta}^{\dagger}(u)$ is
hence
\begin{align}
O' & =\sum_{\vec{p}',\vec{q}'\in\mathcal{S}_{n,k}}\left(\sum_{\vec{p},\vec{q}\in\mathcal{S}_{n,k}}\det[u_{\vec{p}'\vec{p}}]o_{\vec{p},\vec{q}}\det[(u^{\dagger})_{\vec{q}\vec{q}'}]\right)D_{\vec{q}'}^{\vec{p}'}\nonumber \\
 & =\sum_{\vec{p}',\vec{q}'\in\mathcal{S}_{n,k}}\left(U_{k}(u)\cdot o\cdot U_{k}^{\dagger}(u)\right)_{\vec{p}',\vec{q}'}D_{\vec{q}'}^{\vec{p}'}.\label{eq:observable_rotation}
\end{align}
Above, observe that conjugating $O$ by the dimension $\binom{n}{\eta}$
unitary $U_{\eta}(u)$ is equivalent to conjugating $o$ by the much
smaller dimension $\binom{n}{k}$ unitary $U_{k}(u)$. 

\section{\label{sec:rdm_basis}Tomographic completeness}

A well-determined inversion of the measurement channel requires the
choice of random bases to be tomographically complete \citep{Huang2020shadowestimation}.
In other words, any Hermitian $\eta$-particle $n$-mode fermion state
$\rho$ must be shown to be some linear combination
\begin{align}
\rho=\sum_{j}\alpha_{j}U_{\eta}(u_{j}')\Pi_{\vec{z}_{j}}U_{\eta}^{\dagger}(u_{j}')=\sum_{j}\alpha_{j}U_{\eta}(u_{j})\Pi_{[\eta]}U_{\eta}^{\dagger}(u_{j}),\label{eq:density_matrix_sum_of_rotations}
\end{align}
of fermion rotation $U_{\eta}(u_{j}')$ generated by unitaries $u_{j}'$
with coefficients $\alpha_{j}$. Note that we may replace $\Pi_{\vec{z}_{j}}$
with $\Pi_{[\eta]}$, or in fact any diagonal operator without loss
of generality, as there always exists a permutation matrix $v_{\vec{z}_{j}}$
such that $\Pi_{\vec{z}_{j}}=U_{\eta}(v_{\vec{z}_{j}})\Pi_{[\eta]}U_{\eta}^{\dagger}(v_{\vec{z}_{j}})$.
We may then collect $U_{\eta}(u_{j}')U_{\eta}(v_{\vec{z}_{j}})=U_{\eta}(u_{j}'v_{\vec{z}_{j}})=U_{\eta}(u_{j})$
following the homomorphism \ref{eq:homomorphism}. We now prove that
the decomposition of $\rho$ in \ref{eq:density_matrix_sum_of_rotations}
is always achievable.
\begin{thm}[Tomographic completeness]
\label{thm:tomographic_completeness}Any $n$-mode $\eta$-particle
fermion density matrix $\rho$ is a linear combination of diagonal
operators conjugated by some element $U_{\eta}(u)\in\wedge^{\eta}\mathcal{U}_{n}$
as in \ref{eq:density_matrix_sum_of_rotations}.
\end{thm}
\begin{proof}
First, observe that $\rho=\sum_{\vec{p},\vec{q}\in\mathcal{S}_{n,\eta}}\rho_{\vec{p},\vec{q}}\ketbra{\vec{p}}{\vec{q}}$,
where $\ketbra{\vec{p}}{\vec{q}}=\rdm_{\vec{q}}^{\vec{p}}$ forms
a complete basis for $\rho$, and can be written as a sum of Hermitian
operators $\rdm_{\vec{q};\phi}^{\vec{p}}\doteq e^{i\phi}\rdm_{\vec{q}}^{\vec{p}}+e^{-i\phi}\rdm_{\vec{p}}^{\vec{q}}$
like
\begin{align}
\rdm_{\vec{q}}^{\vec{p}} & =\frac{1}{2}\left(\rdm_{\vec{q};0}^{\vec{p}}-i\rdm_{\vec{q};\pi/2}^{\vec{p}}\right),
\end{align}
Hence for any $\rho$, it suffices to show the existence of a decomposition
\begin{equation}
\rdm_{\vec{q};\phi}^{\vec{p}}=\sum_{j}\alpha_{j}U_{\eta}(u_{j})\Pi_{\vec{z}_{j}}U_{\eta}^{\dagger}(u_{j}),\label{eq:hermitian_rdm_lcu}
\end{equation}
for any $\vec{p},\vec{q},\phi$.

Second, it suffices to consider $k$-RDMs $\rdm_{\vec{q}}^{\vec{p}}$
where $\vec{p}$ and $\vec{q}$ have no elements in common. In the
trivial case where $\vec{p}=\vec{q}$, \ref{eq:hermitian_rdm_lcu}
is automatically satisfied as $\rdm_{\vec{q}}^{\vec{p}}=\hat{n}_{\vec{p}}$
is diagonal. Otherwise, if $\vec{p}\cap\vec{q}=\vec{z}$ has more
than zero elements, we may always express $\rdm_{\vec{q}}^{\vec{p}}=(-1)^{x}\hat{n}_{\vec{z}}\rdm_{\vec{q}\backslash\vec{z}}^{\vec{p}\backslash\vec{z}}$
with some sign depending on the order $x$ of the permutation that
anti-commuts out the number operators $\hat{n}_{\vec{z}}$. Now assuming
that $\rdm_{\vec{q}\backslash\vec{z}}^{\vec{p}\backslash\vec{z}}$
is of the form \ref{eq:hermitian_rdm_lcu} where $u_{j}$ is identity
on indices in $\vec{z}$ implies that $\rdm_{\vec{q}}^{\vec{p}}$
also has the same representation as the number operators $\hat{n}_{\vec{z}}$
are already linear combinations of some $\Pi_{\vec{y}}$ and commute
with $U_{\eta}(u_{j})$.

Third, it suffices to restrict the number of particles to $\eta=k$.
Suppose we have found some $\rdm_{\vec{q}}^{\vec{p}}=\sum_{j}\alpha_{j}U_{k}(u_{j})\Pi_{[k]}U_{k}^{\dagger}(u_{j})$
in the form of \ref{eq:hermitian_rdm_lcu}. Note that $\Pi_{[k]}=\hat{n}_{1}\cdots\hat{n}_{k}=D_{[k]}^{[k]}$
is a product of number operators, and is diagonal. Hence \ref{eq:rdm_rotation}
implies that for any $\eta$, $\rdm_{\vec{q}}^{\vec{p}}=\sum_{j}\alpha_{j}U_{\eta}(u_{j})D_{[k]}^{[k]}U_{\eta}^{\dagger}(u_{j})$.

Finally, we show that in the space of $k$-particles, the $k$-RDM
$\rdm_{\vec{q};\phi}^{\vec{p}}$ has a rank-2 decomposition. It also
suffices to assume that $\vec{p}=[k]$ and $\vec{q}=[2k]\backslash[k]=(k+1,\cdots,2k)$
as $\rdm_{\vec{q};\phi}^{\vec{p}}=U_{k}(v)\rdm_{[2k]\backslash[k];\phi}^{[k]}U_{k}^{\dagger}(v)$
for some permutation matrix $v$. To diagonalize the $\rdm_{\vec{q};\phi}^{\vec{p}}$,
observe that it has only two eigenvectors $\ket{\phi_{\pm}}=\frac{1}{\sqrt{2}}\left(e^{i\phi/2}\ket{\vec{p}}\pm e^{-i\phi/2}\ket{\vec{q}}\right)$
with non-zero eigenvalues. Hence
\begin{align}
\rdm_{\vec{q};\phi}^{\vec{p}} & =\ketbra{\phi_{+}}{\phi_{+}}-\ketbra{\phi_{-}}{\phi_{-}}.
\end{align}
Let us find a fermion rotation that rotates $U_{k}(w_{\phi})\ket{\vec{p}}=\ket{\phi_{+}}$
and $U_{k}(w)\ket{\vec{q}}=\ket{\phi_{-}}$. From \ref{eq:fermion_rotation},
we see that $w_{\phi}$ must satisfy the following constraints 
\begin{align}
\frac{1}{\sqrt{2}}\left(e^{i\phi/2}\ket{\eV{\vec{p}}}+e^{-i\phi/2}\ket{\eV{\vec{q}}}\right) & =\sum_{\vec{p}'\in\mathcal{S}_{n,k}}\det\left[(w_{\phi})_{\vec{p}'\vec{p}}\right]\ket{\vec{p}'},\\
\frac{1}{\sqrt{2}}\left(e^{i\phi/2}\ket{\eV{\vec{p}}}-e^{-i\phi/2}\ket{\eV{\vec{q}}}\right) & =\sum_{\vec{q}'\in\mathcal{S}_{n,k}}\det\left[(w_{\phi})_{\vec{q}'\vec{q}}\right]\ket{\vec{q}'}.
\end{align}
As the determinant is zero for linearly dependent columns or rows,
we see by inspection that a solution is
\begin{equation}
w_{\phi}=\frac{1}{\sqrt{2}}\left[\begin{array}{cc}
e^{i\phi/2}I_{k} & e^{i\phi/2}I_{k}\\
e^{-i\phi/2}I_{k} & -e^{-i\phi/2}I_{k}
\end{array}\right],
\end{equation}
where $I_{k}$ the dimension-$k$ identity matrix.
\begin{equation}
\rdm_{[2k]\backslash[k];\phi}^{[k]}=U(w_{\phi})\left(\ketbra{[k]}{[k]}-\ketbra{[2k]\backslash[k]}{[2k]\backslash[k]}\right)U^{\dagger}(w_{\phi}).
\end{equation}

Altogether, we obtain the general case, where there $n$ modes and
$\eta$ particles, the $k$-RDM $\rdm_{\vec{q};\phi}^{\vec{p}}$ with
parameters $\vec{p}\cap\vec{q}=\vec{z}$ corresponding to $\mathrm{\dim(\vec{z})}$
elements in common and $k'=k-\dim(\vec{z})>0$ element not in common,
as follows.
\begin{align}
\rdm_{\vec{q};\phi}^{\vec{p}} & =U(vw_{\phi})\left(\hat{n}_{[k']}-\hat{n}_{[2k']\backslash[k']}\right)\hat{n}_{[2k'+\dim(\vec{z})]\backslash[2k']}U^{\dagger}(vw_{\phi})\nonumber \\
 & =\sum_{j\in\{\vec{p},\vec{q}\}}U(v_{j}w_{\phi})\hat{n}_{[k]}U^{\dagger}(v_{j}w_{\phi})\\
w_{\phi} & =\frac{1}{\sqrt{2}}\left[\begin{array}{cc}
e^{i\phi/2}I_{k'} & e^{i\phi/2}I_{k'}\\
e^{-i\phi/2}I_{k'} & -e^{-i\phi/2}I_{k'}
\end{array}\right]\bigoplus I_{[n]\backslash[2k']},
\end{align}
where $v$ is any permutation that maps $[k']$ to $\vec{p}\backslash\vec{z}$,
$[2k']\backslash[k']$ to $\vec{q}\backslash\vec{z}$, and $[2k'+\dim(\vec{z})]\backslash[2k']$
to $\vec{z}$, and $v_{\vec{p}}$ is any permutation that maps $[k]$
to $\vec{p}$ and $v_{\vec{q}}$ is any permutation that maps $[k]$
to $\vec{q}$. 
\end{proof}
Hence, we conclude that any $k$-RDM 
\begin{align}
\rdm_{\vec{q}}^{\vec{p}} & =\begin{cases}
U_{k}(v)\hat{D}_{[k]}^{[k]}U_{k}^{\dagger}(v), & k'=0,\\
\sum_{j=1}^{4}\alpha_{j}U(w)\hat{D}_{[k]}^{[k]}U^{\dagger}(w), & k'>0,
\end{cases}\label{eq:k-RDM-linear_combination}\\
\alpha & =\frac{1}{2}(1,1,-i,-i),\\
w & =(v_{\vec{p}}w_{0},v_{\vec{q}}w_{0},v_{\vec{p}}w_{\pi/2},v_{\vec{q}}w_{\pi/2}).
\end{align}
is a linear combination of at most four diagonal $k$-RDMs $\hat{D}_{[k]}^{[k]}$
conjugated by single-particle basis rotations $U(w_{j}).$ 

\section{Inverse measurement channel\label{sec:fermion_shadows}}

In this section, we invert the measurement channel with the assurance
from \ref{sec:rdm_basis} that this is possible in principle given
that single-particle rotations are tomographically complete. Following
\ref{eq:density_matrix_sum_of_rotations}, the measurement channel
for any quantum state is
\begin{align}
\mathcal{M}_{\wedge^{\eta}\mathcal{U}_{n}}\left[\rho\right] & =\mathcal{M}_{\wedge^{\eta}\mathcal{U}_{n}}\left[\sum_{j}\alpha_{j}U_{\eta}(u_{j})\Pi_{[\eta]}U_{\eta}^{\dagger}(u_{j})\right]\nonumber \\
 & =\sum_{j}\alpha_{j}U_{\eta}(u_{j})\mathcal{M}_{\wedge^{\eta}\mathcal{U}_{n}}\left[\Pi_{[\eta]}\right]U_{\eta}^{\dagger}(u_{j}).
\end{align}
Above, we apply linearity of the measurement channel, and commute
the fermion rotations out of the measurement channel by a change of
variables in the Haar integral of the twirling channel \ref{eq:twirling_channel}.
Hence, for any state $\rho$, inverting the measurement channel reduces
to evaluating $\mathcal{M}_{\wedge^{\eta}\mathcal{U}_{n}}^{-1}\left[\Pi_{[\eta]}\right]$
on just a single basis state. Below, \ref{thm:inverse_measurement_channel}
evaluates $\mathcal{M}_{\wedge^{\eta}\mathcal{U}_{n}}^{-1}\left[\Pi_{[\eta]}\right]$
in a useful form, first by finding the eigenoperators of $\mathcal{M}_{\wedge^{\eta}\mathcal{U}_{n}}$,
and second by expressing $\Pi_{[\eta]}$ as a linear combination of
these eigenoperators.
\begin{thm}[Inverse measurement channel]
\label{thm:inverse_measurement_channel}For any $\eta$ fermions
on $n$-modes, the inverse measurement channel $\mathcal{M}_{\wedge^{\eta}\mathcal{U}_{n}}^{-1}$
has eigenoperators and eigenvalue
\begin{align}
\forall\vec{x}\cap\vec{y}=\emptyset,\;\tilde{n}_{\vec{x},\vec{y}} & \doteq\prod_{j=1}^{d}\left(\hat{n}_{x_{j}}-\hat{n}_{y_{j}}\right),\\
\mathcal{M}_{\wedge^{\eta}\mathcal{U}_{n}}^{-1}\left[\tilde{n}_{\vec{x},\vec{y}}\right] & =\binom{n+1}{d}\tilde{n}_{\vec{x},\vec{y}}.\label{eq:inverse_measurement_channel_eigenoperator}
\end{align}
On the basis state $\Pi_{[\eta]}$,
\begin{align}
\mathcal{M}_{\wedge^{\eta}\mathcal{U}_{n}}^{-1}\left[\Pi_{[\eta]}\right]= & \sum_{d=0}^{\eta}a_{d}\binom{n+1}{d}\tilde{n}_{d},\label{eq:inverse_measurement_channel_a_RHS}\\
a_{d}\doteq & \frac{(n-2d+1)(n-d-\eta)!(\eta-d)!}{(n-d+1)!},\label{eq:a_coefficients}
\end{align}
where the symmetrized difference operator is
\begin{equation}
\tilde{n}_{d}=\sum_{\vec{x}\in\mathcal{S}_{\eta,d}}\sum_{\vec{y}\in S_{[n]\backslash[\eta],d}}\tilde{n}_{\vec{x},\vec{y}}.\label{eq:symmetrized_difference_operator}
\end{equation}
\end{thm}
\begin{proof}
From \ref{eq:measurement_channel}, the measurement channel becomes
\begin{align}
\mathcal{M}_{\wedge^{\eta}\mathcal{U}_{n}}\left[\Pi_{\vec{p}}\right] & =\binom{n}{\eta}\tr 1{\mathcal{T}_{2,\wedge^{\eta}\mathcal{U}_{n}}\left(\Pi_{\vec{p}}\otimes I\right)},\label{eq:measurement_channel_fermion_rotation}\\
\mathcal{T}_{t,\wedge^{\eta}\mathcal{U}_{n}} & \doteq\int\left(U_{\eta}(u)\Pi_{[\eta]}U_{\eta}^{\dagger}(u)\right)^{\otimes t}\mathrm{d}u_{\mathrm{Haar}}(\mathcal{U}_{n}).
\end{align}
The first step is evaluating the twirling channel from \ref{eq:twirling_channel}.
In the second equality above, we use the fact that $\ketbra{\vec{z}}{\vec{z}}=U_{\eta}(v_{\vec{z}})\Pi_{[\eta]}U_{\eta}^{\dagger}(v_{\vec{z}})$
followed by a change of variables does not affect integration over
the Haar measure. As $\Pi_{\vec{p}}$ is diagonal, it suffices to
evaluate only diagonal components of $\mathcal{T}_{2,\wedge^{\eta}\mathcal{U}_{n}}$.
Each matrix entry of the fermion rotation $U_{\eta}(u)$ is a determinant
of the $\eta\times\eta$ submatrix $u_{\vec{x}\vec{y}}$. By the Leibniz
formula, $\det\left[u_{\vec{x}\vec{y}}\right]$ is a linear combination
of $\eta$ products of entries of $u$. Hence $\mathcal{T}_{t,\wedge^{\eta}\mathcal{U}_{n}}$
is a linear combination of the Haar integral of $2t\eta$ products
of entries from $u$ or $u^{\dagger}$, which are in turn linear combinations
of Weingarten functions \citep{Collins2006Weingarten}. Following
an involved combinatorial proof in \ref{lem:twirling_expression}
of \ref{sec:Haar_integrals}, we find that diagonal components of
the twirling channel are
\begin{align}
\mathcal{T}_{2,\wedge^{\eta}\mathcal{U}_{n}} & =\sum_{\vec{p},\vec{q}\in\mathcal{S}_{n,\eta}}f(\match{\vec{p}}{\vec{q}})\Pi_{\vec{p}}\otimes\Pi_{\vec{q}}+\cdots,\\
f(k) & \doteq\binom{n+1}{\eta}^{-1}\binom{n}{\eta}^{-1}\frac{\eta+1}{\eta+1-k}.
\end{align}

As $f(\match{\vec{p}}{\vec{q}})$ has the same value for all states
with the same number of matching elements, the measurement channel
separates into a sum of projectors 
\begin{align}
\mathcal{M}_{\wedge^{\eta}\mathcal{U}_{n}}\left[\Pi_{\vec{p}}\right] & =\binom{n}{\eta}\sum_{k=0}^{\eta}f(k)\mathrm{Sim}_{k}\left(\Pi_{\vec{p}}\right),\label{eq:measurement_channel_similarity_projector}\\
\mathrm{Sim}_{k}\left(\Pi_{\vec{p}}\right) & \doteq\sum_{\vec{q}\in\mathcal{S}_{n,\eta}:\match{\vec{p}}{\vec{q}}=k}\Pi_{\vec{q}}.\label{eq:sim_k_definition}
\end{align}
In \ref{lem:similarity_projector_to_symmetric_polynomial}, we show
that $\mathrm{Sim}_{k}\left(\Pi_{\vec{p}}\right)=\sum_{j=0}^{\eta}(-1)^{j+k}\binom{j}{k}e_{j}\left(\{\hat{n}\}_{\vec{p}}\right)$
is a linear combination of elementary symmetric polynomials. Substituting
into \ref{eq:measurement_channel_similarity_projector} and applying
the identity $\sum_{k=0}^{j}\frac{(-1)^{j+k}(1+\eta)}{1+\eta-k}\binom{j}{k}=\binom{\eta}{j}^{-1}$,
\begin{equation}
\mathcal{M}_{\wedge^{\eta}\mathcal{U}_{n}}\left[\Pi_{\vec{p}}\right]=\sum_{j=0}^{\eta}\frac{e_{j}\left(\{\hat{n}\}_{\vec{p}}\right)}{\binom{n+1}{\eta}\binom{\eta}{j}}.\label{eq:measurement_channel_symmetric_polynomials}
\end{equation}

This representation allows us to identify the eigenoperators of $\mathcal{M}_{\wedge^{\eta}\mathcal{U}_{n}}$
by applying identity $e_{k}\left(A,X_{2},\cdots\right)-e_{k}\left(B,X_{2},\cdots\right)=(A-B)e_{k-1}\left(X_{2},\cdots\right)$.
As $\Pi_{\vec{p}}=\hat{n}_{p_{1}}\cdots\hat{n}_{p_{\eta}}=e_{\eta}\left(\{\hat{n}\}_{\vec{p}}\right)$,
linearity of the measurement channel implies
\begin{align}
 & \mathcal{M}_{\wedge^{\eta}\mathcal{U}_{n}}\left[\prod_{j=1}^{d}(\hat{n}_{p_{j}}-\hat{n}_{y_{j}})\hat{n}_{p_{d+1}}\cdots\hat{n}_{p_{\eta}}\right]\nonumber \\
 & \quad=\prod_{j=1}^{d}(\hat{n}_{p_{j}}-\hat{n}_{y_{j}})\sum_{j=d}^{\eta}\frac{e_{j-d}\left(\hat{n}_{p_{d+1}},\cdots,\hat{n}_{p_{\eta}}\right)}{\binom{n+1}{\eta}\binom{\eta}{j}},
\end{align}
Let us relabel the variables $\vec{x}=(p_{1},\cdots,p_{d})$, $\vec{z}=(p_{d+1},\cdots,p_{\eta})$.
Note that $\vec{x},\vec{y},\vec{z}$ have no elements in common. Hence,
\begin{equation}
\mathcal{M}_{\wedge^{\eta}\mathcal{U}_{n}}\left[\tilde{n}_{\vec{x},\vec{y}}\hat{n}_{\vec{z}}\right]=\tilde{n}_{\vec{x},\vec{y}}\sum_{j=d}^{\eta}\frac{e_{j-d}\left(\hat{n}_{z_{1}},\cdots,\hat{n}_{z_{\eta-d}}\right)}{\binom{n+1}{\eta}\binom{\eta}{j}}.
\end{equation}
Without loss of generality, let $\vec{x}\cup\vec{y}$ contain the
last $2d$ elements of $[n]$. Hence $\vec{z}\in\mathcal{S}_{n-2d,\eta-d}$.
Now sum both sides over all values of $\vec{z}$ using the summation
identity $\sum_{\vec{q}\in\mathcal{S}_{m,b}}e_{k}\left(X_{q_{1}},\cdots,X_{q_{b}}\right)=\binom{m-k}{b-k}e_{k}\left(X_{1},\cdots,X_{m}\right)$
on both sides. Hence,
\begin{align}
\mathcal{M}_{\wedge^{\eta}\mathcal{U}_{n}}\left[\tilde{n}_{\vec{x},\vec{y}}\right] & =\tilde{n}_{\vec{x},\vec{y}}\sum_{j=d}^{\eta}\frac{\binom{n-d-j}{\eta-j}e_{j-d}\left(\hat{n}_{1},\cdots,\hat{n}_{n-2d}\right)}{\binom{n+1}{\eta}\binom{\eta}{j}}.
\end{align}
Observe that $\tilde{n}_{\vec{x},\vec{y}}\ket{\vec{p}}$ is non-zero
only if $\vec{p}$ contains $d$ elements of $\vec{x}$ and $\vec{y}$.
Hence for any state $\ket{\vec{p}}$ where $\vec{p}$ has $\eta-d$
elements in $[n-2d]$, $e_{j-d}\left(\hat{n}_{1},\cdots,\hat{n}_{n-2d}\right)\ket{\vec{p}}=\binom{\eta-d}{j-d}\ket{\vec{p}}$.
Thus
\begin{align}
\mathcal{M}_{\wedge^{\eta}\mathcal{U}_{n}}\left[\tilde{n}_{\vec{x},\vec{y}}\right] & =\tilde{n}_{\vec{x},\vec{y}}\sum_{j=d}^{\eta}\frac{\binom{n-d-j}{\eta-j}\binom{\eta-d}{j-d}}{\binom{n+1}{\eta}\binom{\eta}{j}}=\frac{\tilde{n}_{\vec{x},\vec{y}}}{\binom{n+1}{d}}.
\end{align}
We then obtain \ref{eq:inverse_measurement_channel_eigenoperator}
by a trivial inversion.

Evaluating $\mathcal{M}_{\wedge^{\eta}\mathcal{U}_{n}}^{-1}\left[\Pi_{[\eta]}\right]$
requires us to express $\Pi_{[\eta]}$ as some linear combination
of $\tilde{n}_{\vec{x},\vec{y}}$. The symmetrized difference operators
$\tilde{n}_{d}$ \ref{eq:symmetrized_difference_operator} are eigenoperators
that enumerate over all combinations of $\vec{x}$ and all permutations
of $\vec{y}$ consistent with $\Pi_{[\eta]}$. Hence, there exists
a linear combination
\begin{align}
\Pi_{[\eta]} & =\sum_{d=0}^{\eta}a_{d}\tilde{n}_{d},\quad a_{d}=\frac{\tr{}{\Pi_{\left\{ 1,\cdots,\eta\right\} }\tilde{n}_{d}}}{\tr{}{\tilde{n}_{d}^{2}}},\label{eq:projector_linear_combination}
\end{align}
where the last line follows from trace orthogonality $\tr{\eta}{\tilde{n}_{j}\tilde{n}_{k}}\propto\delta_{jk}$. 

It is helpful to expand $\tilde{n}_{d}$ in terms of elementary symmetric
polynomials.
\begin{align}
\tilde{n}_{d} & =\sum_{j=0}^{d}(-1)^{j}\binom{d}{j}\frac{\left(\sum_{\vec{z}\in S_{\eta,d}}\right)}{\left(\sum_{\vec{z}\in S_{\eta,d-j}}\right)}\frac{\left(\sum_{\vec{z}\in S_{[n]\backslash[\eta],d}}\right)}{\left(\sum_{\vec{z}\in S_{[n]\backslash[\eta],j}}\right)}\nonumber \\
 & \qquad\times\underbrace{\sum_{\vec{x}\in S_{\eta,d-j}}\hat{n}_{x_{1}}\cdots\hat{n}_{x_{d-j}}}_{e_{d-j}(\hat{n}_{1},\cdots,\hat{n}_{\eta})}\underbrace{\sum_{\vec{y}\in S_{[n]\backslash[\eta],j}}\hat{n}_{y_{1}}\cdots\hat{y}_{j}}_{j!e_{j}(\hat{n}_{\eta+1},\cdots,\hat{n}_{n})},\nonumber \\
 & =\sum_{j=0}^{d}(-1)^{j}\frac{(\eta-d+j)!}{(\eta-d)!}\frac{(n-\eta-j)!}{(n-\eta-d)!}\nonumber \\
 & \qquad\times e_{d-j}(\hat{n}_{1},\cdots,\hat{n}_{\eta})e_{j}(\hat{n}_{\eta+1},\cdots,\hat{n}_{n}).\label{eq:symmetric_difference_elementary_polynomials}
\end{align}
To evaluate $\tr{}{\Pi_{\left\{ 1,\cdots,\eta\right\} }\tilde{n}_{d}}$,
note that $\Pi_{\left\{ 1,\cdots,\eta\right\} }e_{j}(\hat{n}_{\eta+1},\cdots,\hat{n}_{n})\ket{\vec{z}}=\delta_{j,0}\ket{\vec{z}}$
for any basis state with $\eta$ particles and that $\tr{}{\Pi_{\left\{ 1,\cdots,\eta\right\} }e_{d}(\hat{n}_{1},\cdots,\hat{n}_{\eta})}=\binom{\eta}{d}$.
Hence 
\begin{align}
\tr{}{\Pi_{\left\{ 1,\cdots,\eta\right\} }\tilde{n}_{d}} & =\frac{(n-\eta)!}{(n-\eta-d)!}\binom{\eta}{d}.
\end{align}

Evaluating $\tr{}{\tilde{n}_{d}^{2}}$ is significantly more challenging,
and we leave most details to \ref{lem:trace_symmetric_difference_squared}.
There, we show that taking the square of \ref{eq:symmetric_difference_elementary_polynomials}
leads to the sum,
\begin{align}
\tr{}{\tilde{n}_{d}^{2}} & =\frac{\eta!(n-\eta)!(\eta-s)!}{(n-\eta-s)!(\eta-d)!^{2}(n-\eta-d)!^{2}}\label{eq:trace_squared_sum}\\
 & \qquad\times\sum_{s=0}^{\min(\eta,n-\eta)}\left(\sum_{j=\max(0,s+d-\eta)}^{\min(d,s)}f_{\eta,d}(s,j)\right)^{2},\nonumber \\
f_{\eta,d}(s,j) & =\frac{(-1)^{j}(\eta-d+j)!(n-\eta-j)!}{j!(s-j)!(d-j)!(\eta-s-d+j)!},
\end{align}
which evaluates to $\tr{\eta}{\tilde{n}_{d}^{2}}=\frac{\eta!}{d!}\frac{(n-d+1)!(n-\eta)!}{(n-2d+1)(n-\eta-d)!^{2}(\eta-d)!^{2}}$.
Substituting into \ref{eq:projector_linear_combination} and canceling
terms, then leads to \ref{eq:a_coefficients} for $a_{d}$. 

We complete the proof by applying $\mathcal{M}_{\wedge^{\eta}\mathcal{U}_{n}}^{-1}\left[\tilde{n}_{d}\right]=\binom{n+1}{d}\tilde{n}_{d}$
and linearity of measurement channels to $\mathcal{M}_{\wedge^{\eta}\mathcal{U}_{n}}^{-1}\left[\sum_{d=0}^{\eta}a_{d}\tilde{n}_{d}\right]$.
\end{proof}
\begin{lem}
\label{lem:similarity_projector_to_symmetric_polynomial}For any $\vec{p}\in\mathcal{S}_{n,\eta}$,
let the projector onto basis states distance $k$ from $\Pi_{\vec{p}}$
be $\mathrm{Sim}_{k}\left(\Pi_{\vec{p}}\right)$ from \ref{eq:sim_k_definition}.
Then the following equalities are true
\begin{align}
e_{k}\left(\hat{n}_{p_{1}},\cdots\hat{n}_{p_{\eta}}\right) & =\sum_{j=0}^{\eta}\binom{j}{k}\mathrm{Sim}_{j}\left(\Pi_{\vec{p}}\right),\label{eq:ek_as_simj}\\
\mathrm{Sim}_{k}\left(\Pi_{\vec{p}}\right) & =\sum_{j=0}^{\eta}(-1)^{j+k}\binom{j}{k}e_{j}\left(\hat{n}_{p_{1}},\cdots\hat{n}_{p_{\eta}}\right).\label{eq:simj_as_ek}
\end{align}
\end{lem}
\begin{proof}
Without loss of generality, let $\vec{p}=(1,\cdots,\eta)$. For any
integer $s\in[0,\eta]$, let $\vec{r}\in\mathcal{S}_{\eta,s}$, $\vec{q}\in\mathcal{S}_{[n]\backslash[\eta],\eta-s}$.
For any basis state $\ket{\vec{r}\circ\vec{q}}=\ket{(r_{1},\cdots,r_{s},q_{1},\cdots,q_{\eta-s})}$,
\begin{align}
e_{k}\left(\hat{n}_{1},\cdots\hat{n}_{\eta}\right)\ket{\vec{r}\circ\vec{q}} & =\binom{s}{k}\ket{\vec{r}\circ\vec{q}},\\
\mathrm{Sim}_{k}\left(\hat{n}_{1},\cdots\hat{n}_{\eta}\right)\ket{\vec{r}\circ\vec{q}} & =\delta_{k,s}\ket{\vec{r}\circ\vec{q}}.
\end{align}
Hence, the first equality \ref{eq:ek_as_simj} immediately is true
for all basis states $\ket{\vec{r}\circ\vec{q}}$. The second equality
is true as
\begin{align}
\bra{\vec{r}\circ\vec{q}}\mathrm{Sim}_{k}\left(\Pi_{\vec{p}}\right)\ket{\vec{r}\circ\vec{q}} & =\sum_{j=0}^{\eta}(-1)^{j+k}\binom{s}{j}\binom{j}{k}.
\end{align}
When either $s<k$, observe that $\binom{s}{j}\binom{j}{k}=0$. When
$s\ge k$, we apply the fact $\sum_{j=0}^{s}(-1)^{j+k}\binom{s}{j}P(j)=s!a_{s}$
for any polynomial $P(j)=\sum_{x=0}^{s}a_{x}j^{x}$ of degree at most
$s$. Observe that $\binom{j}{k}$ is a polynomial of degree $k$
with coefficient $a_{s}=\frac{1}{s!}\delta_{k,s}$. Hence $\sum_{j=0}^{s}(-1)^{j+k}\binom{s}{j}\binom{j}{s}=\delta_{k,s}.$
As $s\le\eta$, we may change the upper limits of the sum to $\sum_{j=0}^{\eta}\cdots$
without affecting its value.
\end{proof}

\section{\label{sec:efficient_estimation_from_shadows}Efficient estimation
from fermion shadows}

In the previous section, we found an expression for the inverse measurement
channel, which provides an unbiased single-shot estimate $\hat{\rho}_{u,\vec{z}}$
of the quantum state $\rho$. Hence, from \ref{eq:state_estimator},
the estimator for any observable $O$ is
\begin{equation}
\langle\hat{O}\rangle=\tr{}{OU_{\eta}^{\dagger}(u)\mathcal{M}_{\wedge^{\eta}\mathcal{U}_{n}}^{-1}\left[\Pi_{\vec{z}}\right]U_{\eta}(u)}.\label{eq:O_naive_estimator}
\end{equation}
However, this expression is not computationally efficient as the right-hand
side of \ref{eq:inverse_measurement_channel_a_RHS} for $\mathcal{M}_{\wedge^{\eta}\mathcal{U}_{n}}^{-1}\left[\Pi_{[n]}\right]$
is a sum of exponentially many terms, and each operator in the trace
above has exponentially large dimension $\binom{n}{\eta}\times\binom{n}{\eta}$.
However, we now show below in \ref{thm:single_shot_estimate_efficient}
that for observables that are linear combination of $k$-RDMs, as
in \ref{eq:observable_linear_combination_rdm}, the estimator \ref{eq:O_naive_estimator}
simplifies to multiplying at most $4$ dimension $\binom{n}{k}\times\binom{n}{k}=\mathcal{O}(n^{2k})$
matrices, and is hence efficient for any fixed $k$. We then go further
in \ref{sec:Pffaffians} to show that efficient evaluation with respect
to $k$ is also possible by modifying a recent technique based on
polynomial interpolation of Pfaffians \citep{Wan2022MatchgateShadows}.

Furthermore in the special case where $O=D_{\vec{q}}^{\vec{p}}$ is
just a single $k$-RDM, the single shot estimate \ref{thm:single_shot_estimate_efficient}
simplifies to 
\begin{equation}
\langle\hat{D}_{\vec{q}}^{\vec{p}}\rangle=\bra{\vec{q}}U_{k}^{\dagger}(v_{\vec{z}}^{\dagger}u)E_{\eta,k}U_{k}(v_{\vec{z}}^{\dagger}u)\ket{\vec{p}}.\label{eq:single-shot}
\end{equation}
Hence the single evaluation of
\begin{equation}
U_{k}^{\dagger}(v_{\vec{z}}^{\dagger}u)E_{\eta,k}U_{k}(v_{\vec{z}}^{\dagger}u)=\sum_{\vec{p},\vec{q}\in\mathcal{S}_{n,k}}\ketbra{\vec{q}}{\vec{p}}\langle\hat{D}_{\vec{q}}^{\vec{p}}\rangle,
\end{equation}
simultaneously estimates all $k$-RDMs.
\begin{thm}[Estimator for $k$-RDMs]
\label{thm:single_shot_estimate_efficient}Let $(u,\vec{z})$ be
a classical shadow of $\rho$. Let the observable $O=\sum_{\vec{p},\vec{q}\in\mathcal{S}_{n,k}}o_{\vec{p},\vec{q}}\rdm_{\vec{q}}^{\vec{p}}.$
Then the single-shot estimator is
\begin{align}
\langle\hat{O}\rangle & =\tr{}{o\cdot U_{k}^{\dagger}(v_{\vec{z}}^{\dagger}u)\cdot E_{\eta,k}\cdot U_{k}(v_{\vec{z}}^{\dagger}u)},\label{eq:single_shot_estimate_thm_equation}\\
E_{\eta,k} & \doteq\sum_{\vec{r}\in\mathcal{S}_{n,k}}\ketbra{\vec{r}}{\vec{r}}\frac{\binom{\eta-s'}{k-s'}\binom{n-\eta+s'}{s'}}{(-1)^{k+s'}\binom{k}{s'}},\label{eq:estimation_matrix}
\end{align}
where $s'=\left|\vec{r}\cap[\eta]\right|$ and $v_{\vec{z}}$ is any
permutation that maps elements of $[\eta]$ to $\vec{z}$.
\end{thm}
\begin{proof}
The inverse measurement channel on $\Pi_{\vec{z}}=U_{\eta}^{\dagger}(v_{\vec{z}}^{\dagger})\Pi_{[\eta]}U_{\eta}(v_{\vec{z}}^{\dagger})$
can be expressed as one on $\Pi_{[\eta]}$ through
\begin{equation}
\mathcal{M}_{\wedge^{\eta}\mathcal{U}_{n}}^{-1}\left[\Pi_{\vec{z}}\right]=U_{\eta}^{\dagger}(v_{\vec{z}}^{\dagger})\mathcal{M}_{\wedge^{\eta}\mathcal{U}_{n}}^{-1}\left[\Pi_{[\eta]}\right]U_{\eta}(v_{\vec{z}}^{\dagger}).
\end{equation}
The estimator \ref{eq:O_naive_estimator} combined with the expression
$\mathcal{M}_{\wedge^{\eta}\mathcal{U}_{n}}^{-1}\left[\Pi_{[\eta]}\right]$
from \ref{thm:inverse_measurement_channel}, is
\begin{align}
\langle\hat{O}\rangle & =\sum_{d=0}^{\eta}a_{d}\binom{n+1}{d}\tr{}{U_{\eta}(v_{\vec{z}}^{\dagger}u)OU_{\eta}^{\dagger}(v_{\vec{z}}^{\dagger}u)\tilde{n}_{d}}.
\end{align}
To simplify notation, let $w=v_{\vec{z}}^{\dagger}u$. As $\tilde{n}_{d}$
is diagonal, only the diagonal components of $U_{\eta}(w)OU_{\eta}^{\dagger}(w)$
have non-zero contributions to the trace. When $O$ from \ref{eq:observable_linear_combination_rdm}
is a sum of $k$-RDMs, \ref{eq:observable_rotation} tells that diagonal
components of
\begin{equation}
U_{\eta}(w)OU_{\eta}^{\dagger}(w)=\sum_{\vec{r}\in\mathcal{S}_{n,k}}\left(U_{k}(w)\cdot o\cdot U_{k}^{\dagger}(w)\right)_{\vec{r},\vec{r}}\hat{n}_{\vec{r}}+\cdots.
\end{equation}
Hence, the single-shot estimate is
\begin{align}
\langle\hat{O}\rangle & =\tr{}{U_{k}(w)\cdot o\cdot U_{k}^{\dagger}(w)\cdot E_{\eta,k}},\\
E_{\eta,k}= & \sum_{\vec{r}\in\mathcal{S}_{n,k}}\ketbra{\vec{r}}{\vec{r}}\tr{}{\hat{n}_{\vec{r}}\tilde{n}_{d}}.
\end{align}
Using the cyclic property of the trace, this matches the form of \ref{eq:single_shot_estimate_thm_equation}. 

However, this evaluation is still not efficient as the estimation
matrix $E_{\eta,k}$ is still a sum over exponentially many terms
through $\tilde{n}_{d}=\sum_{\vec{x}\in S_{\eta,d}}\sum_{\vec{y}\in S_{[n]\backslash[\eta],d}}\tilde{n}_{\vec{x},\vec{y}}$.
The sum over permutations implies that the trace has a value that
depends on $\vec{r}$ only through the number of elements that are
not in $[\eta]$. Suppose $\hat{n}_{\vec{r}}$ has $k-s=\left|\vec{r}\cap\text{\ensuremath{[\eta]}}\right|$
elements that overlap with $[\eta]$ and $s=\left|\vec{r}\backslash\text{\ensuremath{[\eta]}}\right|$
elements that overlap with $[n]\backslash[\eta]$. Then let
\begin{align}
E_{\eta,k}= & \sum_{s=0}^{k}E_{\eta,k,s}\sum_{\vec{r}\in\mathcal{S}_{n,k}:\left|\vec{r}\backslash\text{\ensuremath{[\eta]}}\right|=s}\ketbra{\vec{r}}{\vec{r}},\label{eq:estimation_matrix_subspaces}\\
E_{\eta,k,s} & \doteq\tr{}{\hat{n}_{\vec{r}}\tilde{n}_{d}}=\tr{}{\hat{n}_{[k-s]}\hat{n}_{\eta+[s]}\tilde{n}_{d}}.
\end{align}
By expressing $\tilde{n}_{d}$ in terms of elementary symmetric polynomials
\ref{eq:symmetrized_difference_operator}, we obtain
\begin{align}
E_{\eta,k,s} & =\sum_{d'=0}^{d}\frac{(\eta-d+d')!}{(-1)^{d'}(\eta-d)!}\frac{(n-\eta-d')!}{(n-\eta-d)!}E_{\eta,k,s,d'}\\
E_{\eta,k,s,d'} & =\mathrm{Tr}\big[\hat{n}_{[k-s]}e_{d-d'}(\hat{n}_{1},\cdots,\hat{n}_{\eta})\hat{n}_{\eta+[s]}\nonumber \\
 & \qquad\times e_{d'}(\hat{n}_{\eta+1},\cdots,\hat{n}_{n})\big].
\end{align}
Now, observe that $\hat{n}_{j}^{2}=\hat{n}_{j}$ and that some of
the number operators in $\hat{n}_{[k-s]}\hat{n}_{\eta+[s]}$ may be
identical to some of the terms in the symmetric polynomials. Thus
we separate the sums in $e_{d-d'}(\cdots)$ and $e_{d'}(\cdots)$
into cases where the number operator indices contain $x'$ and $y'$
elements of $[k-s]$ and $\eta+[s]$ respectively, that is
\begin{align}
E_{\eta,k,s,d'} & =\sum_{x'=0}^{k-s}\sum_{y'=0}^{s}\binom{k-s}{x'}\binom{\eta-k+s}{d-d'-x'}\binom{s}{y'}\\
 & \qquad\times\binom{n-\eta-s}{d'-y'}\binom{n-(d+k-x'-y')}{\eta-(d+k-x'-y')}.\nonumber 
\end{align}
Combining all these expressions, 
\begin{equation}
E_{\eta,k,s}=\sum_{d=0}^{\eta}\sum_{d'=0}^{d}\sum_{x'=0}^{k-s}\sum_{y'=0}^{s}\cdots,\label{eq:estimation_matrix_entries}
\end{equation}
is a quadruple sum over hypergeometric terms that we simplify in \ref{lem:estimation_matrix_values}
to obtain \ref{eq:estimation_matrix}.
\end{proof}

\subsection{\label{sec:Pffaffians}$k$-RDM estimator in matrix multiplication
time}

We now present a scheme for computing the estimator for any $k$-RDM
in the time $\mathcal{O}(k^{2}\eta^{\alpha}+\eta k)$ time, where
multiplying a $k\times\eta$ matrix with a $\eta\times k$ matrix
takes $\mathcal{O}(k^{2}\eta^{\alpha})$ time for some matrix multiplication
exponent $\alpha\le1$. From each measurement of the quantum state,
we can estimate any $k$-RDM in $\mathcal{O}(k^{2}\eta)$ time, which
is significant improvement, especially in applications \citep{Huggins2022FermionicQMC}
where these estimates must be computed a large number of times.

Following \ref{sec:rdm_basis} and \ref{eq:single-shot}, computing
the estimate of any $k$-RDM $\langle\hat{D}_{\vec{q}}^{\vec{p}}\rangle$
from a classical shadow $(u,\vec{z})$ reduces to computing the estimate
at most four diagonal $k$-RDMs of the form 
\begin{align}
\langle\hat{n}_{[k]}\rangle_{v_{\vec{z}}uw_{j}}\doteq & \langle U_{k}(w_{j})\hat{D}_{[k]}^{[k]}U_{k}^{\dagger}(w_{j})\rangle\nonumber \\
= & \tr k{\hat{n}_{[k]}U_{k}^{\dagger}(v_{\vec{z}}^{\dagger}uw_{j})E_{\eta,k}U_{k}(v_{\vec{z}}^{\dagger}uw_{j})},
\end{align}
where the trace is over the $k$-particle subspace, i.e. $\tr kX=\sum_{\vec{x}\in\mathcal{S}_{n,k}}\bra{\vec{x}}X\ket{\vec{x}}=\tr{}{\Pi_{k}X},$where
$\Pi_{k}$ projects onto $k$-particle subspace. For brevity, let
$u=v_{\vec{z}}^{\dagger}uw_{j}$. From \ref{eq:estimation_matrix_subspaces},
the estimation matrix from \ref{eq:estimation_matrix} has the same
value for all $\vec{r}$ that share the same value $s=\left|\vec{r}\backslash\text{\ensuremath{[\eta]}}\right|$.
Hence, we may write
\begin{align}
E_{\eta,k} & =\sum_{s=0}^{k}\underbrace{E_{\eta,k,s}'}_{E_{\eta,k,k-s}}\mathrm{Sim}_{k,s}\left(\Pi_{[\eta]}\right),\\
\mathrm{Sim}_{k,s}\left(\Pi_{[\eta]}\right) & \doteq\sum_{\vec{r}\in\mathcal{S}_{n,k}:|\vec{r}\cap[\eta]|=s}\ketbra{\vec{r}}{\vec{r}}\nonumber \\
 & =\sum_{j=0}^{k}(-1)^{j+s}\binom{j}{s}e_{j}\left(\{\hat{n}\}_{[\eta]}\right),\label{eq:sim_k_s_number_operators}
\end{align}
where the proof for the second equality of $\mathrm{Sim}_{k,s}\left(\Pi_{[\eta]}\right)$
is similar to \ref{lem:similarity_projector_to_symmetric_polynomial}
for the special case $\mathrm{Sim}_{k}\left(\Pi_{[\eta]}\right)=\mathrm{Sim}_{\eta,k}\left(\Pi_{[\eta]}\right)$.
Then estimating $\langle\hat{D}_{\vec{q}}^{\vec{p}}\rangle=\tr k{D_{\vec{q}}^{\vec{p}}U_{k}^{\dagger}(u)E_{\eta,k}U_{k}(u)}$
reduces to estimating terms of the form

\begin{equation}
\langle\hat{n}_{[k]}\rangle_{u}=\sum_{s=0}^{k}E'_{\eta,k,s}\tr k{\hat{n}_{[k]}U_{k}^{\dagger}(u)\mathrm{Sim}_{k,s}\left(\Pi_{[\eta]}\right)U_{k}(u)}.
\end{equation}

We proceed further by representing regular fermions operators with
Majorana fermion $\gamma_{p}$ operators where
\begin{align}
a_{p} & =\frac{\gamma_{2p-1}+i\gamma_{2p}}{2},\quad a_{p}^{\dagger}=\frac{\gamma_{2p-1}-i\gamma_{2p}}{2},\\
\hat{n}_{j} & =\frac{1}{2}\big(1+i\beta_{j}\big),\quad\beta_{j}\doteq\gamma_{2j-1}\gamma_{2j},
\end{align}
and $\beta_{j}$ is a Majorana bivector. In this representation, the
single-particle rotations $U(u)$ rotates Majorana operators to other
linear combination of Majorana operators
\begin{align}
&U(u)\gamma_{p}U^{\dagger}(u)  =\sum_{q}\tilde{u}_{p,q}\gamma_{q},\\
&\tilde{u}  =\mathrm{Re}[u]\otimes\underbrace{\left(\begin{array}{cc}
1 & 0\\
0 & 1
\end{array}\right)}_{I_{[2]}}+\mathrm{Im}[u]\otimes\underbrace{\left(\begin{array}{cc}
0 & -1\\
1 & 0
\end{array}\right)}_{Y}.
\end{align}
More generally, given any real orthogonal matrix $R\in\mathcal{O}_{2n}$,
let us overload the notation for $U$ with
\begin{align}
U(R)\gamma_{p}U^{\dagger}(R) & =\sum_{q}R_{q,p}\gamma_{q}.
\end{align}
Elementary symmetric polynomials of number operators $e_{j'}\left(\{\hat{n}\}_{[k]}\right)$
can be shown to be linear combinations of $e_{j}\left(\{\beta\}_{[\eta]}\right)$
for $j'\in[0,j]$. From the generating function
\begin{equation}
\prod_{j=0}^{\eta}(\kappa-\lambda\hat{n}_{j})=\sum_{j=0}^{\eta}\kappa^{\eta-j}(-\lambda)^{j}e_{j}\left(\{\hat{n}\}_{[\eta]}\right),
\end{equation}
 for elementary symmetric polynomials, one can show that
\begin{align}
&\prod_{j=0}^{\eta}\left(\kappa-\frac{\lambda}{2}-i\frac{\lambda}{2}\beta_{j}\right)\nonumber \\ &\qquad =\sum_{\alpha=0}^{\eta}\left(\kappa-\frac{\lambda}{2}\right)^{\eta-\alpha}\left(-i\frac{\lambda}{2}\right)^{\alpha}e_{\alpha}\left(\{\beta\}_{[\eta]}\right).
\end{align}
The coefficient of the monomial $(-\lambda)^{j}$ is, after setting
$\kappa=1$, and using the binomial theorem,
\begin{align}
e_{j}\left(\{\hat{n}\}_{[\eta]}\right) & =\frac{1}{2^{j}}\sum_{\alpha=0}^{\eta}\binom{\eta-\alpha}{j-\alpha}i^{\alpha}e_{\alpha}\left(\{\beta\}_{[\eta]}\right).\label{eq:symmetric_n_to_beta}
\end{align}
Substituting \ref{eq:symmetric_n_to_beta} into \ref{eq:sim_k_s_number_operators},
we obtain 
\begin{align}
\mathrm{Sim}_{k,s}\left(\Pi_{[\eta]}\right) & =(-1)^{s}\sum_{j=0}^{\eta}f_{k,s}(j)i{}^{j}e_{j}\left(\{\beta\}_{[\eta]}\right),\label{eq:Sim_k_s_symmetric_polynomial_beta}\\
f_{k,s}(j) & \doteq\sum_{x=j}^{k}(-1)^{x}\binom{x}{s}\frac{1}{2^{j}}\binom{\eta-j}{\eta-x}.
\end{align}

The estimate $\langle\hat{\Pi}_{[k]}\rangle_{u}$ is then a linear
combination traces of various degrees of elementary symmetric polynomials
of Majorana bivectors
\begin{align}
\langle\hat{n}_{[k]}\rangle_{u} & =\sum_{j=0}^{k}\alpha_{\eta,k,j}\tr k{\hat{n}_{[k]}U_{k}^{\dagger}(u)e_{j}\left(\{\beta\}_{[\eta]}\right)U_{k}(u)},\label{eq:estimator_symmetric-polynomials}\\
\alpha_{\eta,k,j} & \doteq\sum_{s=0}^{k}(-1)^{s}f_{k,s}(j)i{}^{j}E'_{\eta,k,s}.
\end{align}
 with some coefficients $\alpha_{\eta,k,j}$. The speedup we present
in this section arises from an improved method for computing in $\mathcal{O}(\eta k^{2})$
time traces of the form
\begin{equation}
\tr k{\hat{n}_{[k]}U(u)e_{j}\left(\{\beta\}_{[\eta]}\right)U^{\dagger}(u)}.\label{eq:inner_product}
\end{equation}
simultaneously for all elementary symmetric polynomials $j=0,\cdots,\eta$
of Majorana bivectors. Once computed, the $k$-RDM estimate $\langle U_{k}(u)\hat{D}_{[k]}^{[k]}U_{k}^{\dagger}(u)\rangle$
is obtained by an appropriate linear combination of \ref{eq:inner_product}. 

We note that inner product \ref{eq:inner_product} can be inferred
from $\mathcal{O}(n)$ calculations of fermionic linear optics \citep{Bravyi2005},
which is related to the strong simulability of matchgate circuits
or non-interacting fermion distributions \citep{Terhal2002fermions},
and is known to be possible in polynomial time $\mathcal{O}(\text{poly}(n))$
\citep{Brod2016Matchgate}. Given any real anti-symmetric matrix $M$
that has block-diagonal form 
\begin{equation}
M=R\left[\bigoplus_{j\in[n]}\lambda_{j}Y\right]R^{T},\quad R\in\mathcal{SO}_{2n},
\end{equation}
let the fermionic gaussian \citep{Bravyi2005}
\begin{align}
\rho(M) & =e^{\frac{i}{2}\sum_{j,k\in[2n]}\theta_{j}M_{j,k}\theta_{k}}\\\nonumber &=U(R)\prod_{j\in[n]}(1+i\lambda_{j}\beta_{j})U^{\dagger}(R).\label{eq:fermionic_gaussian}
\end{align}
Following \citep{Bravyi2005}, it is well-known that the trace of
two products of fermionic gaussians evaluates to a Pfaffian as follows
\begin{align}
\tr{}{\rho(M)\rho(K)} & =\mathrm{Pf}(M)\mathrm{Pf}(K-M^{-1}).
\end{align}

Our algorithm works in two key steps. First, it reduces the inner
product computation with each elementary symmetric polynomials $e_{x}(\beta_{1},\cdots,\beta_{\eta})$
\ref{eq:inner_product} to evaluating at most $\eta$ instances of
fermionic linear optics $\tr{}{\rho(M)\rho(K)}$ for some $M$ and
$K$. Second, theses $\eta$ instances turn out to be related and
may actually all be evaluated with one unit of effort.

In the first step, we show that for any operator $X$, the trace of
$\tr k{\hat{n}_{[k]}X}=\tr k{\rho_{k}X}$ over $k$ particles is equal
to the trace over all particles times an appropriate fermionic gaussian
\begin{align}
\rho_{k} & \doteq\rho\left(\Lambda\otimes Y\right)=\hat{n}_{[k]}\prod_{j\in[n]\backslash[k]}(1-\hat{n}_{j}),\\
\Lambda & \doteq I_{[k]}-I_{[n]\backslash[k]}=I_{[n]}-2I_{[n]\backslash[k]}=2I_{[k]}-I_{[n]},
\end{align}
for some invertible matrix $\Lambda$. This implies the desired equality
\begin{align}
\tr k{\hat{n}_{[k]}X} & =\tr{}{\hat{n}_{[k]}\prod_{j\in[n]\backslash[k]}(1-\hat{n}_{j})X}=\tr{}{\rho_{k}X}.
\end{align}

Next, observe that the fermionic gaussian itself is a generating function
for elementary symmetric polynomials 
\begin{align}
\rho(\kappa M) & =U(R)\prod_{j=1}^{n}(1+i\kappa\lambda_{j}\beta_{j})U^{\dagger}(R) \\ \nonumber
 & =U(R)\sum_{j=0}^{n}(i\kappa)^{j}e_{j}\left(\lambda_{1}\beta_{1},\cdots,\lambda_{n}\beta_{n}\right)U^{\dagger}(R),
\end{align}
Hence the following fermonic gaussian is a linear combination of the
symmetric polynomials $e_{j}\left(\{\beta\}_{[\eta]}\right)$ seen
in \ref{eq:Sim_k_s_symmetric_polynomial_beta}
\begin{equation}
\rho_{\eta}(\kappa)\doteq\rho\left(\kappa I_{[\eta]}\otimes Y\right)=\sum_{j=0}^{\eta}(i\kappa)^{j}e_{j}\left(\{\beta\}_{[\eta]}\right).
\end{equation}
Thus, a single evaluation of 
\begin{align}
 & \tr{}{\rho_{k}U_{k}(u)\rho_{\eta}(\kappa)U_{k}^{\dagger}(u)}\nonumber\\
 & \quad=\sum_{j=0}^{\eta}(i\kappa)^{j}\tr k{U(u)\hat{n}_{[k]}U^{\dagger}(u)e_{j}\left(\{\beta\}_{[\eta]}\right)},
\end{align}
is some linear combination of the desired quantity \ref{eq:inner_product}.
Evaluating this on $\mathcal{O}(\eta)$ different values of $\kappa$
then provides enough information to compute all $\tr k{\Pi_{[k]}U(v)e_{j}\left(\{\hat{n}\}_{[\eta]}\right)U^{\dagger}(v)}$,
e.g. by polynomial interpolation.

In the second step, we avoid using polynomial interpolation and find
a faster approach. Observe that the traces we compute turn out to
be Pfaffians of appropriately defined matrices.
\begin{align}
 & \tr{}{\rho_{k}U_{k}(u)\rho_{\eta}(\kappa)U_{k}^{\dagger}(u)} =(-1)^{n-k}\mathrm{Pf}\left[A(\kappa)\right].\label{eq:pfaffian}\\
& A(\kappa) \doteq\kappa I_{[\eta]}\otimes Y-\tilde{u}^{T}\Lambda\otimes Y\tilde{u}
\end{align}
By taking high-order derivatives with respect to $\kappa$, we are
able to isolate the traces with individual elementary symmetric polynomials.
For instance,
\begin{align}
 & \partial_{\kappa}^{x}\tr{}{\rho_{k}U_{k}(u)\rho_{\eta}(\kappa)U_{k}^{\dagger}(u)}\vert_{\kappa=0}\nonumber \\
 & \quad=x!i^{x}\tr k{\hat{n}_{[k]}U(u)e_{x}\left(\{\beta\}_{[\eta]}\right)U^{\dagger}(u)}.
\end{align}
In the following, we present an efficient method to compute all derivatives
of the Pfaffian in \ref{eq:pfaffian}. Observe that the derivatives
of a Pfaffian in general is
\begin{align}
\partial\mathrm{Pf}(A) & =\frac{1}{2}\mathrm{Pf}(A)\tr{}{A^{-1}\partial A}.
\end{align}
In our case, the derivatives are\begin{widetext}
\begin{align}
\partial_{\kappa}\mathrm{Pf}(A(\kappa)) & =\frac{1}{2}\mathrm{Pf}(A)\tr{}{A^{-1}(I_{[\eta]}\otimes Y)},\\
\partial_{\kappa}^{2}\mathrm{Pf}(A(\kappa)) & =\frac{1}{2}\partial_{\kappa}\mathrm{Pf}(A)\tr{}{A^{-1}(I_{[\eta]}\otimes Y)}-\frac{1}{2}\mathrm{Pf}(A)\tr{}{A^{-1}(I_{[\eta]}\otimes Y)A^{-1}(I_{[\eta]}\otimes Y)},\\
\vdots & =\nonumber \\
\partial_{\kappa}^{x}\mathrm{Pf}(A(\kappa)) & =\frac{1}{2}\sum_{j=0}^{x-1}\binom{x-1}{j}\partial_{\kappa}^{x-j-1}(-1)^{j}\mathrm{Pf}(A)\tr{}{\left(A^{-1}(I_{[\eta]}\otimes Y)\right)^{j+1}}.
\end{align}
\end{widetext}This recursion allows us to compute higher-order derivatives
from lower-order derivatives. After computing all the traces $\tr{}{\left(A^{-1}(I_{[\eta]}\otimes Y)\right)^{j}}\vert_{\kappa=0}$,
the recursion for all derivatives $\forall x\in[\eta],\;\partial_{\kappa}^{x}\mathrm{Pf}(A)\vert_{\kappa=0}$
can be solved in $\mathcal{O}(\eta^{2})$ time.

We now evaluate the trace. Observe that
\begin{align}
A(0) & =A^{-1}(0)=-\tilde{u}^{T}\Lambda\otimes Y\tilde{u}.
\end{align}
Hence, the trace
\begin{align}
 & \tr{}{\left(A^{-1}(I_{[\eta]}\otimes Y)\right)^{j}}\vert_{\kappa=0}\nonumber \\
 & =(-1)^{j}\tr{}{\left((\Lambda\otimes Y)(\tilde{u}I_{[\eta]}\otimes Y\tilde{u}^{T})\right)^{j}}.
\end{align}
In principle, it suffices to evaluate the all eigenvalues of $(\Lambda\otimes Y)(\tilde{u}I_{[\eta]}\otimes Y\tilde{u}^{T})$.
Computing eigenvalues of a $2n\times2n$ matrix takes $\mathcal{O}(n^{\omega})$
time and would enable the straightforward computation of the trace.
Indeed, a similar approach was taken in \citep{Wan2022MatchgateShadows}
for the non-particle conserving case. However, we now highlight optimizations
for our particle-conserving case that reduces the problem to finding
the eigenvalues of an even smaller $2k\times2k$ matrix. Using the
identity $\Lambda=2I_{[k]}-I_{[n]}$,
\begin{widetext}
\begin{align}
 & \tr{}{\left(A^{-1}(I_{[\eta]}\otimes Y)\right)^{j}}\vert_{\kappa=0}\nonumber \\
 & =(-1)^{j}\tr{}{\left(2(\tilde{u}^{T}I_{[k]}\otimes Y\tilde{u})(I_{[\eta]}\otimes Y)-(\tilde{u}^{T}I_{[n]}\otimes Y\tilde{u})(I_{[\eta]}\otimes Y)\right)^{j}}\nonumber \\
 & =(-1)^{j}\tr{}{\left(-2(\tilde{iu}^{T}I_{[k]}\otimes I_{[2]}\tilde{iu})(I_{[\eta]}\otimes I_{[2]})+I_{[\eta]}\otimes I_{[2]}\right)^{j}},
\end{align}
\end{widetext}
where $\tilde{iu}\doteq-\mathrm{Im}[u]\otimes I_{[2]}+\mathrm{Re}[u]\otimes Y.$
Hence using the binomial expansion and the cyclic property of the
trace,\begin{widetext}
\begin{align}
\tr{}{\left(A^{-1}(I_{[\eta]}\otimes Y)\right)^{j}}\vert_{\kappa=0} & =(-1)^{j}\sum_{y=0}^{j}(-2)^{y}\binom{j}{y}\tr{}{I_{[\eta]}\otimes I_{[2]}\left[(\tilde{iu}^{T}I_{[k]}\otimes I_{[2]}\tilde{iu})(I_{[\eta]}\otimes I_{[2]})\right]^{y}}\nonumber \\
 & =(-1)^{j}\left[\tr{}{I_{[\eta]}\otimes I_{[2]}}+\sum_{y=1}^{j}(-2)^{y}\binom{j}{y}\tr{}{\left[(\tilde{iu}^{T}I_{[k]}\otimes I_{[2]}\tilde{iu})(I_{[\eta]}\otimes I_{[2]})\right]^{y}}\right]\nonumber \\
 & =(-1)^{j}\left[2\eta+\sum_{y=1}^{j}(-2)^{y}\binom{j}{y}\tr{}{\left[\underbrace{(I_{[k]}\otimes I_{[2]})(\tilde{iu}I_{[\eta]}\otimes I_{[2]}\tilde{iu}^{T})(I_{[k]}\otimes I_{[2]})}_{M}\right]^{y}}\right].
\end{align}
\end{widetext}

Within the trace, observe that the matrix $M$ is non-zero on only
a $2k\times2k$ block. We may write $M=m\cdot m^{T}$ where 
\begin{equation}
m=(I_{[k]}\otimes I_{[2]})(\tilde{iu}I_{[\eta]}\otimes I_{[2]}),
\end{equation}
 is a $2k\times2\eta$ submatrix of $\tilde{iu}$. We note that constructing
the $2k\times2\eta$ submatrix of $\tilde{iu}$ for any $u=v_{\vec{z}}^{\dagger}uv_{\vec{p}}w_{\phi}$
only takes $\mathcal{O}(k\eta)$ arithmetic opertations if $v_{\vec{z}}^{\dagger},v_{\vec{p}},w_{\phi}$,
which contain at most two non-zero elements in any row or column,
are stored as sparse matrices. Hence $M$ is obtained by matrix multiplication
of a $2k\times2\eta$ matrix with a $2\eta\times2k$ matrix, which
takes $\mathcal{O}(k^{2}\eta^{\alpha})$ time for some matrix multiplication
exponent $\alpha\le1$. Once $M$ is obtained, we may obtain all of
its non-zero eigenvalues in $\mathcal{O}(k^{\omega})$ times. Subsequently,
all powers of its eigenvalues up to $\eta$ may be obtained (such
as by repeated squaring) in $\mathcal{O}(\eta k)$ multiplications,
which allows all traces $\tr{}{M^{y}}$ for all $y\in[\eta]$ to be
computed in an overall time of $\mathcal{O}(k^{2}\eta^{\alpha}+k^{\omega}+\eta k)$.
Assuming the naive cubic-time algorithm for matrix this takes $\mathcal{O}(k^{2}\eta)$
times, dominated by the cost of forming $M$. Subsequently, all traces
$\tr{}{\left(A^{-1}(I_{[\eta]}\otimes Y)\right)^{j}}\vert_{\kappa=0}$
for all $j\in[\eta]$ may be computed recursively in $\mathcal{O}(\eta^{2})$
time, which implies an overall complexity of $\mathcal{O}(\eta^{2}+k^{2}\eta)$
for computing all derivatives of the Pfaffian $\partial_{\kappa}^{x}\mathrm{Pf}(A)$.

\section{Error of estimation\label{sec:shadow_norm}}

We now evaluate the variance of our estimator from \ref{sec:efficient_estimation_from_shadows}
for any $k$-RDM. For brevity, we use the notation for the expectation
$\mathbb{E}_{u}\doteq\mathbb{E}_{u\sim\mathcal{U}_{n}}$ and $\mathbb{E}_{u,\vec{z}}\doteq\mathbb{E}_{u\sim\mathcal{U}_{n}}\mathbb{E}_{\ket{\vec{z}}\sim U_{\eta}(u)\rho U_{\eta}^{\dagger}(u)}\doteq\mathbb{E}_{u\sim\mathcal{U}_{n}}\sum_{\vec{z}}U_{\eta}(u)\rho U_{\eta}^{\dagger}(u)\ket{\vec{z}}$.
Consider an observable $O=O_{\mathrm{tr}}+\alpha I$ where $O_{\mathrm{tr}}$
is traceless. From the definition of variance for any estimate $\langle\hat{O}\rangle=\tr{}{O\hat{\rho}_{u,\vec{z}}}$,
the variance
\begin{align}
\mathrm{Var}\left[\langle\hat{O}\rangle\right] & =\mathbb{E}_{u,\vec{z}}\left[\left|\tr{}{O\hat{\rho}_{u,\vec{z}}}-\tr{}{O\rho}\right|^{2}\right]\nonumber \\
 & =\mathbb{E}_{u,\vec{z}}\left[\left|\tr{}{O_{\mathrm{tr}}\hat{\rho}_{u,\vec{z}}}\right|^{2}\right]-\left|\tr{}{O_{\mathrm{tr}}\rho}\right|^{2},
\end{align}
only depends on the traceless component. An upper bound on the variance
is then the state-dependent shadow norm
\begin{equation}
\norm O_{\text{s,\ensuremath{\rho}}}^{2}\doteq\mathbb{E}_{u,\vec{z}}\left[\left|\tr{}{O_{\mathrm{tr}}\hat{\rho}_{u,\vec{z}}}\right|^{2}\right].
\end{equation}

As the traceless component of any $k$-RDM is
\begin{equation}
(D_{\vec{q}}^{\vec{p}})_{\mathrm{tr}}=D_{\vec{q}}^{\vec{p}}-\delta_{\vec{p},\vec{q}}\frac{\binom{n-k}{\eta-k}}{\binom{n}{\eta}}I,
\end{equation}
the shadow norm of any $k$-RDM is then
\begin{align}
 & \norm{D_{\vec{q}}^{\vec{p}}}_{\text{s,\ensuremath{\rho}}}^{2}=\mathbb{E}_{u,\vec{z}}\left[\left|\langle\hat{D}_{\vec{q}}^{\vec{p}}\rangle\right|^{2}\right]\nonumber \\
 & \qquad-\delta_{\vec{p},\vec{q}}\left(2\frac{\binom{n-k}{\eta-k}\mathbb{E}_{u,\vec{z}}\left[\langle\hat{D}_{\vec{q}}^{\vec{p}}\rangle\right]}{\binom{n}{\eta}}-\frac{\binom{n-k}{\eta-k}^{2}}{\binom{n}{\eta}^{2}}\right),\label{eq:k_rdm_shadow_norm}
\end{align}
where $\langle\hat{D}_{\vec{q}}^{\vec{p}}\rangle=\tr{}{D_{\vec{q}}^{\vec{p}}\hat{\rho}_{u,\vec{z}}}$
is our single-shot estimator from \ref{eq:single_shot_estimate_thm_equation}
in the previous section that implicitly depends on the shadow $(u,\vec{z})$.
One may also define a state-independent shadow norm by maximizing
$\norm O_{\text{s}}^{2}\doteq\max_{\rho}\norm O_{\text{s,\ensuremath{\rho}}}^{2}.$
Below, we prove that the average variance of estimating all $k$-RDMs
is also a quantity independent of $\rho$ and also small.
\begin{thm}[Error of estimation]
\label{thm:error_bound}For all $\eta$-particle $n$-mode fermion
states $\rho$, the average variance over all $k$-RDMs is upper-bounded
by the average squared shadow norm
\begin{align}
\mathbb{E}_{\vec{p},\vec{q}}\left[\norm{D_{\vec{q}}^{\vec{p}}}_{\text{s},\rho}^{2}\right] & =\frac{\tr{}{E_{\eta,k}^{2}}}{\binom{n}{k}^{2}}-\frac{\binom{n-k}{\eta-k}^{2}}{\binom{n}{\eta}^{2}\binom{n}{k}}\label{eq:shadow_norm_expectation_pq}\\
 & \le\binom{\eta}{k}\left(1-\frac{\eta-k}{n}\right)^{k}\left(\frac{1+n}{1+n-k}\right).\nonumber 
\end{align}
\end{thm}
\begin{proof}
Consider the case $\vec{p}=\vec{q}$, where the cross-term $\mathbb{E}_{u,\vec{z}}\left[\langle\hat{D}_{\vec{p}}^{\vec{p}}\rangle\right]=\mathbb{E}_{u,\vec{z}}\left[\tr{}{\hat{n}_{\vec{p}}\hat{\rho}_{u,\vec{z}}}\right]$
in \ref{eq:k_rdm_shadow_norm} appears. Using the identity $\sum_{\vec{p}\in\mathcal{S}_{n,k}}\hat{n}_{\vec{p}}=\sum_{\vec{p}\in\mathcal{S}_{n,k}}e_{k}(\hat{n}_{p_{1}},\cdots,\hat{n}_{p_{k}})=e_{k}(\hat{n}_{1},\cdots,\hat{n}_{n})$,
observe that on the space of $\eta$-particle states, 
\begin{equation}
e_{k}(\hat{n}_{1},\cdots,\hat{n}_{n})=\frac{\binom{n-k}{\eta-k}\binom{n}{k}}{\binom{n}{\eta}}I.
\end{equation}
Hence, the sum 
\begin{align}
K\doteq\sum_{\vec{p}\in\mathcal{S}_{n,k}}\mathbb{E}_{u,\vec{z}}\left[\langle\hat{D}_{\vec{p}}^{\vec{p}}\rangle\right] & =\frac{\binom{n-k}{\eta-k}\binom{n}{k}}{\binom{n}{\eta}},
\end{align}
and the sum of shadow norms over all diagonal $k$-RDMs 
\begin{align}
\sum_{\vec{p}\in\mathcal{S}_{n,k}}\norm{D_{\vec{p}}^{\vec{p}}}_{\text{s,\ensuremath{\rho}}}^{2} & =\sum_{\vec{p}\in\mathcal{S}_{n,k}}\mathbb{E}_{u,\vec{z}}\left[\left|\langle\hat{D}_{\vec{p}}^{\vec{p}}\rangle\right|^{2}\right]-K,
\end{align}
and the sum of shadow norms over all $k$-RDMs 
\begin{align}
 & \sum_{\vec{p},\vec{q}\in\mathcal{S}_{n,k}}\norm{D_{\vec{q}}^{\vec{p}}}_{\text{s},\rho}^{2}=\mathbb{E}_{u,\vec{z}}\left[\sum_{\vec{p},\vec{q}\in\mathcal{S}_{n,k}}\left|\langle\hat{D}_{\vec{q}}^{\vec{p}}\rangle\right|^{2}\right]-K\nonumber \\
 & \qquad=\mathbb{E}_{u,\vec{z}}\left[\tr{}{\left(U_{k}^{\dagger}(v_{\vec{z}}u)E_{\eta,k}U_{k}(v_{\vec{z}}u)\right)^{2}}\right]-K\nonumber \\
 & \qquad=\tr{}{E_{\eta,k}^{2}}-\frac{\binom{n-k}{\eta-k}^{2}\binom{n}{k}}{\binom{n}{\eta}^{2}},\label{eq:sum_pq_shadow_norm}
\end{align}
is state-independent, where in the second line, we substitute the
estimator \ref{eq:single_shot_estimate_thm_equation}, and in the
last line, we use the cyclic property of traces to cancel all adjacent
unitaries $U_{k}(v_{\vec{z}}u)U_{k}^{\dagger}(v_{\vec{z}}u)=I$. We
then obtain \ref{eq:shadow_norm_expectation_pq} by diving \ref{eq:sum_pq_shadow_norm}
by the number of terms $\binom{n}{k}^{2}$ in the sum.

From \ref{eq:sum_pq_shadow_norm}, an upper bound on the average shadow
norm is just $\mathbb{E}_{\vec{p},\vec{q}}\left[\norm{D_{\vec{q}}^{\vec{p}}}_{\text{s},\rho}^{2}\right]\le\frac{\tr{}{E_{\eta,k}^{2}}}{\binom{n}{k}^{2}}\doteq Q_{n,\eta,k}$.
From \ref{thm:single_shot_estimate_efficient}, $\bra{\vec{p}}E_{\eta,k}\ket{\vec{p}}=(-1)^{s}\binom{k}{s}^{-1}\binom{\eta-k+s}{s}\binom{n-\eta+k-s}{k-s}$,
where $s=\left|\vec{p}\backslash\text{\ensuremath{[\eta]}}\right|$.
By writing the sum $\sum_{\vec{p}\in\mathcal{S}_{n,k}}\cdots=\sum_{s=0}^{k}\binom{n-\eta}{s}\binom{\eta}{k-s}\cdots$
and collecting terms, 
\begin{align}
 & Q_{n,\eta,k}=\binom{\eta}{k}\sum_{s=0}^{k}\binom{k}{s}\left[\prod_{j=0}^{k-1}\frac{n-\eta+k-s-j}{n-j}\right]\nonumber \\
 & \qquad\times\left\{ \frac{(\eta-k+s)!(n-\eta+k-s)!(n-k)!}{n!(\eta-k)!(n-\eta)!}\right\} .
\end{align}
We arrive at the bound in \ref{eq:shadow_norm_expectation_pq} by
observing that the product on the right $\left[\prod_{j=0}^{k-1}\cdots\right]\le\left(1-\frac{\eta-k}{n}\right)^{k}$,
and that $\sum_{s=0}^{k}\binom{k}{s}\left\{ \cdots\right\} =\frac{1+n}{1+n-k}$. 
\end{proof}
We note that tighter bounds may be derived for specific parameter
regimes. For instance, when $k=\eta$, we may increase the upper limit
on the sum as the summand is zero when $s\ge\eta$. Hence, the sum
\begin{align}
Q_{n,\eta,\eta} & =\sum_{s=0}^{n}\frac{\eta!(n-\eta)!}{(\eta-s)!(n-s-\eta)!}\frac{(n-s)!^{2}}{n!^{2}}\label{eq:slater_error}\\
 & \le Q_{2\eta,\eta,\eta}=\sum_{s=0}^{2\eta}\prod_{j=0}^{s-1}\frac{(\eta-j)^{2}}{(2\eta-j)^{2}}\le\sum_{s=0}^{2\eta}\frac{1}{4^{s}}\le\frac{4}{3},\nonumber 
\end{align}
is symmetric about and maximized at $\eta=n/2$. Other interesting
asymptotic bounds include the limit of large $n$ with $\eta$ held
constant. There $\left(1-\frac{\eta-k}{n}\right)^{k}=1-\mathcal{O}(\frac{k(\eta-k)}{n})$
implies $Q_{n,\eta,k}=\binom{\eta}{k}$$\left(1-\mathcal{O}\left(\frac{k(\eta-k)}{n}\right)\right)$.
At half-filling $\eta=n/2$, in the limit of large $n$ with $k$
held constant, $\left(1-\frac{\eta-k}{n}\right)^{k}=2^{-k}\text{\ensuremath{\left(1+\mathcal{O}(k^{2}/n)\right)}}$
implies $Q_{n,n2,k}=\binom{\eta}{k}2^{-k}\left(1+\mathcal{O}(k^{2}/n)\right)$.

\section{Conclusion\label{sec:Conclusion}}

We have presented a technique to estimate all $k$-RDMs of fermion
states extremely efficiently with an average error that depends only
on the number of particles, in contrast to all prior methods which
depend on the number of modes. Our main assumption that the state
of interest has a definite number of particles applies to very many
systems of interest. On a quantum computer, our scheme may be applied
with any fermion-to-qubit mapping so long as the random number-conserving
single-particle rotations can be applied. Implementing our approach
in the second-quantized representation also facilitates a straightforward
approach to error mitigation by symmetry verification \citep{Endo2018ErrorMitigation,Bonet2018ErrorMitigation}
and can detect multiple errors -- simply check that the measured
state $\ket{\vec{z}}$ has $\eta$-particles.

Our estimator is also computationally efficient for all parameters.
Many observables such as electronic structure or nuclear Hamiltonians
only require $k\le3$. However, there are natural applications for
large $k$ as well. Consider the problem of estimating overlaps with
arbitrary Slater determinants $\ket{\vec{q}'}=U_{\eta}(w)\ket{\vec{q}}$
for any $w\in\mathcal{U}_{n}$, a key component of quantum-classical
auxiliary-field quantum Monte Carlo \citep{Huggins2022FermionicQMC}.
For any pure state $\ket{\psi}=\sum_{\vec{p}\in\mathcal{S}_{n,\eta}}\psi_{\vec{p}}\ket{\vec{p}'}$,
add $\eta$ more modes and prepare the state
\begin{equation}
\ket{\psi'}=\frac{\ket{\psi}+\ket{n+[\eta]}}{\sqrt{2}}.\label{eq:slater_state}
\end{equation}
Then the $\eta$-RDM $U_{\eta}(w)D_{\vec{q}}^{n+[\eta]}U_{\eta}^{\dagger}(w)$
has expectation
\begin{equation}
\tr{}{U_{\eta}(w)D_{\vec{q}}^{n+[\eta]}U_{\eta}^{\dagger}(w)\ketbra{\psi'}{\psi'}}=\frac{\psi_{\vec{q}'}}{2},
\end{equation}
which is half of the desired overlap with the Slater determinant $\ket{\vec{q}'}$.
Our approach simultaneously estimates all $\eta$-RDMs extremely efficiently
using only $\frac{4}{3\epsilon^{2}}$ samples on average according
to \ref{eq:slater_error}, which compares favorably to very recent
work \citep{Wan2022MatchgateShadows} that performs the same task
using exponentially more $\tilde{\mathcal{O}}(\sqrt{n}/\epsilon^{2})$
samples, through with a stronger per-RDM error guarantee rather than
an average. Many variations on this idea are possible future directions
to pursue. For instance, the number of additional modes required may
also be reduced to as few as $1$ by preparing $\ket{\psi'}=\frac{\ket{\psi}+\ket{n-\eta+1+[\eta]}}{\sqrt{2}}$
and estimating $D_{\vec{q}}^{n-\eta+1+[\eta]}U_{\eta}^{\dagger}(w)$,
or even to $0$ when the sign is not important by estimating $U_{\eta}(w)D_{\vec{q}}^{\vec{q}}U_{\eta}^{\dagger}(w)$.

As our estimator has no preferred basis, any rotated $k$-RDM $U_{\eta}(w)D_{\vec{q}}^{\vec{p}}U_{\eta}^{\dagger}(w)$
may be estimated just as efficiently and easily as $D_{\vec{q}}^{\vec{p}}$
even if they contain exponentially many terms in the computational
basis. Even more general $k$-RDMs of the form $U_{\eta}(w)D_{\vec{q}}^{\vec{p}}U_{\eta}^{\dagger}(w')$
where $w\neq w'$ may also be estimated following \ref{eq:single_shot_estimate_thm_equation},
though we leave the Pfaffian method in \ref{sec:Pffaffians} for this
case to a future analysis. This suggests that our approach is tailored
to estimating observables of the form $O=\sum_{j\in[R]}\alpha_{j}U_{\eta}(w_{j})D_{\vec{q}_{j}}^{\vec{p}_{j}}U_{\eta}(w_{j}')$,
similar to \ref{eq:density_matrix_sum_of_rotations}, that either
have a low-rank structure or are well-approximated by it, meaning
that $R$ is small and $\left|\vec{\alpha}\right|_{2}$ is minimized.
As the upper-bound on the variance of our estimator was evaluated
using the $2$-fold twirling channel, this approach bounded the average
variance across all $k$-RDMs rather than each $k$-RDM individually.
We leave to future work the task of obtaining the shadow norm of linear
combinations of $k$-RDMs such as $O$, which would require a challenging
evaluation of the $3$-fold twirling operator for the group $\wedge^{\eta}\mathcal{U}_{n}$,
followed by understanding the covariance of $k$-RDM estimation.

\begin{acknowledgments}

We thank Aarthi Sundaram, Jeongwan Haah, Matthias Troyer, and Dave
Wecker for insightful discussions.\end{acknowledgments}

\onecolumngrid

\bibliography{main}

\appendix

\section{\label{sec:Haar_integrals}Haar integrals over $\wedge^{\eta}\mathcal{U}_{n}$}

Integrals of polynomial functions with respect to the Haar measure
on the dimension $n$ unitary group $\mathcal{U}_{n}$ can be evaluated
as follows. 
\begin{defn}[Twirling operator]
\label{def:twirling_operator}The twirling operator of degree $t$
on $\bigwedge_{k=1}^{\eta}\mathcal{U}_{n}$ with respect to the Haar
measure on $\mathcal{U}_{n}$ is 
\begin{align}
\mathcal{T}_{t,\mathcal{U}} & \doteq\int\left(U(u)\ket{[\eta]}\bra{[\eta]}U^{\dagger}(u)\right)^{\otimes t}\mathrm{d}u_{\mathrm{Haar}}(\mathcal{U}_{n}),
\end{align}
where $\ket{\vec{\eta}}=\ket{[\eta]}\doteq\bigwedge_{k=1}^{\eta}\ket{\eV k}$
is a basis vector of $\bigwedge_{k=1}^{\eta}\mathcal{U}_{n}$ and
$\ket{\eV k}$ are orthonormal bases for $\mathcal{U}_{n}$.

The purpose of this section is to prove the main result we use for
the measurement channel. In the following, we will choose randomization
with respect to single-particle basis rotations $U_{\eta}(u)$, where
$u\in\mathcal{U}_{n}$ is Haar random. In other words, $\mathcal{U}=\wedge^{\eta}\mathcal{U}_{n}$.
The corresponding twirling operators are linear in the basis of Weingarten
integrals, which will be key to evaluating difficult quantities such
as the measurement channel and shadow variance without the help of
$t$-design results.
\end{defn}
\begin{thm}
\label{lem:twirling_expression}Twirling operator
\begin{equation}
\mathcal{T}_{2,\wedge^{\eta}\mathcal{U}_{n}}=\sum_{\vec{p}\in\mathcal{S}_{n,\eta}^{\otimes2}}f(\match{\vec{p}_{1}}{\vec{p}_{2}})\Pi_{\vec{p}_{1}}\otimes\Pi_{\vec{p}_{2}}+\sum_{\vec{p}\neq\vec{q}\in\mathcal{S}_{n,\eta}^{\otimes2}}f(\vec{p},\vec{q})\ket{\eV{\vec{p}_{1}}}\bra{\eV{\vec{q}_{1}}}\otimes\ket{\eV{\vec{p}_{2}}}\bra{\eV{\vec{q}_{2}}},
\end{equation}
where for the diagonal terms, the structure factor $f(k)=\binom{n+1}{\eta}^{-1}\binom{n}{\eta}^{-1}\frac{\eta+1}{\eta+1-k}$,
and the form of $f(\vec{p},\vec{q})$, of lesser interest, is detailed
in \ref{lem:twirling_expression_full_proof}.
\end{thm}
Our derivation begins by expressing the twirling operator in the basis
of Weingarten integrals. 
\begin{defn}[{Weingarten integral \citep[Equation (1)]{Collins2006Weingarten}}]
\label{def:Weingarten_integral}The basis $(\vec{i},\vec{i}',\vec{j},\vec{j}')$
Weingarten integral of degree $t$ on $\mathcal{U}_{n}$ is
\begin{align}
\int_{u\sim\mathcal{U}_{n}}u_{i_{1}j_{1}}\cdots u_{i_{q}j_{q}}u_{j'_{1}i'_{1}}^{\dagger}\cdots u_{j'_{q}i'_{q}}^{\dagger}\mathrm{d}u & =\sum_{\pi,\xi\in S_{q}}\delta_{\vec{i},\pi(\vec{i}')}\delta_{\xi^{-1}(\vec{j}),\vec{j}'}\wg{\pi\xi^{-1}},\\
\delta_{\vec{x},\pi(\vec{x}')} & \doteq\begin{cases}
1, & \forall j\in[q],\;x_{j}=x'_{\pi(j)},\\
0, & \text{otherwise}.
\end{cases}
\end{align}
where $S_{q}$ is is the symmetric group on $q$ elements, and $\wg{\pi\xi^{-1}}=\wg{\xi^{-1}\pi}=\wg{\pi^{-1}\xi}$
is the so-called Weingarten function, which depends only on the conjugacy
class of the permutation. 
\end{defn}
With the help of Weingarten integrals, we may evaluate the twirling
operator and express then in a simpler form. We use the notation $\vec{x}\oplus\vec{y}=(x_{1},..,x_{\dim(\vec{x})},y_{1},\cdots,y_{\dim(\vec{y})})$
for list concatenation, and $S_{k}$ for the symmetric group on $k$
elements.
\begin{lem}[Twirling operator structure factor]
\label{lem:twirling_expression_full_proof}\label{thm:twirling_operator_diagonal}The
twirling operator evaluates to
\begin{align}
\mathcal{T}_{t,\wedge^{\eta}\mathcal{U}_{n}} & =\sum_{\vec{p}_{\theta},\vec{q}_{\theta}\in\mathcal{S}_{n,\eta}}f(\vec{p},\vec{q})\bigotimes_{\theta\in[t]}\ket{\eV{\vec{p}_{\theta}}}\bra{\eV{\vec{q}_{\theta}}},\\
f(\vec{p},\vec{q}) & =\sum_{\mu\in S_{\eta}^{\oplus t}}\sum_{\nu\in S_{t}^{\oplus\eta}}\sum_{\xi\in S_{t\eta}}(-1)^{\mu}\prod_{\theta\in[t]}\det\left[\Delta_{\vec{q}_{\theta},\vec{p}_{\theta}^{\xi}}\right]\wg{\mu\nu\xi},
\end{align}
where $f$ is the structure factor, $\vec{p}\doteq\bigoplus_{j\in[t]}\vec{p}_{j}$,
$\vec{q}\doteq\bigoplus_{j\in[t]}\vec{q}_{j}$, $\Delta_{ij}=\delta_{ij}$,
and $\bigoplus_{j\in[t]}\vec{p}_{j}^{\xi}\doteq\xi\left(\vec{p}\right),$$\quad\mu=\bigoplus_{j\in[t]}\mu_{j},\quad\nu=\bigoplus_{j\in[\eta]}\nu_{j}$
\begin{equation}
\mu\left(\bigoplus_{\theta\in[t]}\vec{x}_{\theta}\right)=\left(\begin{array}{ccc}
x_{1,\mu_{1}(1)} & \cdots & x_{t,\mu_{1}(1)}\\
\vdots & \ddots & \vdots\\
x_{1,\mu_{\eta}(\eta)} & \cdots & x_{t,\mu_{\eta}(\eta)}
\end{array}\right),\quad\nu\left(\bigoplus_{j\in[t]}\vec{x}_{j}\right)\doteq\left(\begin{array}{ccc}
x_{\nu_{1}(1),1} & \cdots & x_{\nu_{t}(t),1}\\
\vdots & \ddots & \vdots\\
x_{\nu_{1}(1),\eta} & \cdots & x_{\nu_{t}(t),\eta}
\end{array}\right).
\end{equation}
\end{lem}
\begin{proof}
The twirling operator from \ref{def:twirling_operator} is $\mathcal{T}_{t,\wedge^{\eta}\mathcal{U}_{n}}\doteq\int_{u\sim\mathcal{U}_{n}}\left(U(u)\ket{\vec{\eta}}\bra{\vec{\eta}}U^{\dagger}(u)\right)^{\otimes t}\mathrm{d}u$.
By expanding the fermion rotation using \ref{eq:fermion_rotation},
\begin{align}
\mathcal{T}_{t,\wedge^{\eta}\mathcal{U}_{n}} & =\sum_{\vec{p}_{\theta}.\vec{q}_{\theta}\in\mathcal{S}_{n,\eta}}\underbrace{\int_{u\sim\mathcal{U}_{n}}\prod_{\theta\in[t]}\det\left[u_{\vec{q}_{\theta}\vec{\eta}}\right]\det\left[u_{\vec{\eta}\vec{p}_{\theta}}^{\dagger}\right]\mathrm{d}u}_{f(\vec{p},\vec{q})}\bigotimes_{\theta\in[t]}\ket{\eV{\vec{p}_{\theta}}}\bra{\eV{\vec{q}_{\theta}}}.
\end{align}
Note that transposing a matrix does not change its determinant. We
now expand the determinants into Weingarten integrals \ref{def:Weingarten_integral}
using the Leibniz formula
\begin{align}
\det\left[u_{\vec{\eta}\vec{x}}^{\dagger}\right] & =\sum_{\tau\in S_{\eta}}(-1)^{\tau}\prod_{i\in[\eta]}\left(u_{\tau(\vec{\eta})\vec{x}}^{\dagger}\right)_{ii}=\sum_{\tau\in S_{\eta}}(-1)^{\tau}\prod_{i\in[\eta]}u_{\tau(i)x_{i}}^{\dagger}.\label{eq:leibnizdeterminant}\\
\det\left[u_{\vec{x}\vec{\eta}}\right] & =\sum_{\sigma\in S_{\eta}}(-1)^{\sigma}\prod_{i\in[\eta]}\left(u_{\sigma(\vec{x})\vec{\eta}}\right)_{ii}=\sum_{\sigma\in S_{\eta}}(-1)^{\sigma}\prod_{i\in[\eta]}u_{x_{\sigma(i)}i},
\end{align}
where $(-1)^{\sigma}=\mathrm{sgn}(\sigma)$ is parity of the permutation
$\sigma$. Above, we use the notation $\left(\sigma(\vec{x})\right)_{j}=x_{\sigma(j)}=j$
for any vector $\vec{x}\in\mathbb{Z^{\eta}}$. In particular, $\left(\sigma(\vec{\eta})\right)_{j}=\eta_{\sigma(j)}=\sigma(j)$
since $\vec{\eta}=\left(1,2,\cdots,\eta\right)$. The following identity
will be useful
\begin{align}
\det\left[\Delta_{\vec{x},\vec{y}}\right] & =\sum_{\sigma\in S_{\eta}}(-1)^{\sigma}\delta_{\vec{x},\sigma(\vec{y})}=\sum_{\sigma\in S_{\eta}}(-1)^{\sigma}\prod_{i\in[\eta]}\delta_{x_{i},y_{\sigma(i)}}=\begin{cases}
(-1)^{\sigma}, & \vec{y}=\sigma(\vec{x}),\\
0, & \text{otherwise}.
\end{cases}
\end{align}
Hence the structure factor
\begin{align}
f(\vec{p},\vec{q}) & \doteq\int_{u\sim\mathcal{U}_{n}}\sum_{\sigma_{\theta},\tau_{\theta}\in S_{\eta}}\mathrm{sgn}(\underbrace{\oplus_{j\in[t]}\sigma_{j}}_{\sigma}\underbrace{\oplus_{j\in[t]}\tau_{j}}_{\tau})\prod_{\theta\in[t]}\prod_{i\in[\eta]}u_{q_{\theta,\sigma_{\theta}(i)},i}\prod_{i\in[\eta]}u_{\tau_{\theta}(i),p_{\theta,i}}^{\dagger}\mathrm{d}u\nonumber \\
 & =\sum_{\pi,\xi\in S_{t\eta}}\sum_{\sigma_{\theta},\tau_{\theta}\in S_{\eta}}\mathrm{sgn}(\sigma\tau)\delta_{\sigma(\vec{q}),\xi\left(\vec{p}\right)}\delta_{\vec{\eta}^{\oplus t},\pi\tau\left(\vec{\eta}^{\oplus t}\right)}\wg{\xi\pi^{-1}}\nonumber \\
 & =\sum_{\pi,\xi\in S_{t\eta}}\sum_{\tau_{\theta}\in S_{\eta}}\mathrm{sgn}(\tau)\prod_{\theta\in[t]}\det\left[\Delta_{\vec{q}_{\theta},\vec{p}_{\theta}^{\xi}}\right]\delta_{\pi\left(\vec{\eta}^{\oplus t}\right),\tau\left(\vec{\eta}^{\oplus t}\right)}\wg{\xi\pi}\nonumber \\
 & =\sum_{\pi,\xi\in S_{t\eta}}\prod_{\theta\in[t]}\det\left[\Delta_{\vec{q}_{\theta},\vec{p}_{\theta}^{\xi}}\right]\det\left[\Delta_{\vec{\eta},\vec{\eta}_{\theta}^{\pi}}\right]\wg{\xi\pi},
\end{align}
Above, we use the notation $\bigoplus_{\theta\in[t]}\vec{\eta}_{\theta}^{\pi}\doteq\pi\left(\vec{\eta}^{\oplus t}\right)$
and $\bigoplus_{\theta\in[t]}\vec{p}_{\theta}^{\xi}\doteq\xi\left(\vec{p}\right)$. 

We may simplify further as $\det\left[\Delta_{\vec{\eta},\vec{\eta}_{\theta}^{\pi}}\right]$
is non-zero only when $\pi\left(\vec{\eta}^{\oplus t}\right)=\bigoplus_{\theta\in[t]}\,\mu_{\theta}(\vec{\eta})$
is a permutation that preserves the same set of terms $\left\{ k\right\} _{k\in[\eta]}$
for each $j$. The number of unique valid $\pi$ is thus $\eta!\binom{t}{1}^{\eta}\eta!\binom{t-1}{1}^{\eta}\cdots\eta!\binom{1}{1}^{\eta}=\left(\eta!\right)^{t}\left(t!\right)^{\eta}$.
Hence valid $\pi$ decompose into $\pi=\mu\nu$, where $\nu\in S_{t}^{\oplus\eta}$
which swaps elements at the same position between different copies
of $\vec{\eta}$, and $\mu\in S_{\eta}^{\oplus t}$ permutes within
each $\vec{\eta}$. More precisely,
\begin{equation}
\mu\left(\bigoplus_{\theta\in[t]}\vec{x}_{\theta}\right)=\left(\begin{array}{ccc}
x_{1,\mu_{1}(1)} & \cdots & x_{t,\mu_{1}(1)}\\
\vdots & \ddots & \vdots\\
x_{1,\mu_{\eta}(\eta)} & \cdots & x_{t,\mu_{\eta}(\eta)}
\end{array}\right),\quad\nu\left(\bigoplus_{j\in[t]}\vec{x}_{j}\right)\doteq\left(\begin{array}{ccc}
x_{\nu_{1}(1),1} & \cdots & x_{\nu_{t}(t),1}\\
\vdots & \ddots & \vdots\\
x_{\nu_{1}(1),\eta} & \cdots & x_{\nu_{t}(t),\eta}
\end{array}\right).\label{eq:permutation_factorization}
\end{equation}
This decomposition exactly characterizes valid permutations without
any over-counting as $|S_{\eta}^{\oplus t}||S_{t}^{\oplus\eta}|=\left(t!\right)^{\eta}\left(\eta!\right)^{t}$.
Thus
\begin{align}
f(\vec{p},\vec{q}) & =\sum_{\nu\in S_{t}^{\oplus\eta}}\sum_{\mu\in S_{\eta}^{\oplus t}}\sum_{\xi\in S_{t\eta}}(-1)^{\mu}\prod_{\theta\in[t]}\det\left[\Delta_{\vec{q}_{\theta},\vec{p}_{\theta}^{\xi}}\right]\wg{\xi\mu\nu},
\end{align}
and we then apply the cyclic property $\wg{\xi\mu\nu}=\wg{\mu\nu\xi}$.
\end{proof}
We find for the structure factor $f(\vec{p},\vec{q})$ that the case
where $\vec{q}$ is some permutation of $\vec{p}$ occurs quite frequently.
We may further simplify the twirling operator in this case.
\begin{lem}[Twirling operator diagonal structure factor]
\label{lem:twirling_expression_diagonal_full_proof}The coefficient
of diagonal components of the twirling operator 
\begin{align}
\mathcal{T}_{t,\wedge^{\eta}\mathcal{U}_{n}} & =\sum_{\vec{p}\in\mathcal{S}_{n,\eta}^{\otimes t}}f(\vec{p})\bigotimes_{\theta\in[t]}\ket{\eV{\vec{p}_{\theta}}}\bra{\eV{\vec{p}_{\theta}}}+\sum_{\vec{p}\neq\vec{q}\in\mathcal{S}_{n,\eta}^{\otimes t}}f(\vec{p},\vec{q})\bigotimes_{\theta\in[t]}\ket{\eV{\vec{p}_{\theta}}}\bra{\eV{\vec{q}_{\theta}}},
\end{align}
is the diagonal structure factor
\begin{equation}
f(\vec{p})=\sum_{\nu\in S_{\vec{\eta},\cdots,\vec{\eta}}}\sum_{\gamma\in S_{\vec{p}_{1},\cdots,\vec{p}_{t}}}\sum_{\mu,\chi\in S_{\eta}^{\oplus t}}(-1)^{\mu\chi}\wg{\mu\nu\chi\gamma}.
\end{equation}
where $\nu\in S_{\vec{\eta},\cdots,\vec{\eta}}$ swaps matching elements
between $t$ copies of $\eta$, and $\gamma\in S_{\vec{p}_{1},\cdots,\vec{p}_{t}}$
swaps matching elements between the $t$ vectors $\vec{p}_{j}$.
\end{lem}
\begin{proof}
Observe that the twirling operator has unit trace, $\tr{}{\mathcal{T}_{t,\wedge^{\eta}\mathcal{U}_{n}}}=1$.
Using the expression for it from \ref{lem:twirling_expression_full_proof},
\begin{align}
\tr{}{\mathcal{T}_{t,\wedge^{\eta}\mathcal{U}_{n}}} & =\sum_{\vec{p}_{\theta}\in\mathcal{S}_{n,\eta}}f(\vec{p},\vec{p}).
\end{align}
Let us define $f(\vec{p})\doteq f(\vec{p},\vec{p})$. Hence
\begin{equation}
f(\vec{p})=\sum_{\mu\in S_{\eta}^{\oplus t}}\sum_{\nu\in S_{t}^{\oplus\eta}}\sum_{\xi\in S_{t\eta}}(-1)^{\mu}\prod_{\theta\in[t]}\det\left[\Delta_{\vec{p}_{\theta},\vec{p}_{\theta}^{\xi}}\right]\wg{\mu\nu\xi}.
\end{equation}
Now observe that $\det\left[\Delta_{\vec{p}_{\theta},\vec{p}_{\theta}^{\xi}}\right]$
is non-zero only when $\bigoplus_{\theta\in[t]}\vec{p}_{\theta}=\xi\left(\bigoplus_{\theta\in[t]}\vec{p}_{\theta}\right)$
is a permutation that preserves the same set of terms in $\vec{p}_{\theta}$.
Similar to the derivation of \ref{eq:permutation_factorization},
valid $\xi$ decomposes into $\xi=\chi\gamma$, $\gamma\in S_{\vec{p}_{1},\cdots,\vec{p}_{t}}$
which swaps matching elements between different $\vec{p}_{j}$, and
$\chi\in S_{\eta}^{\oplus t}$ permutes within each $\vec{p}_{j}$.
Thus
\begin{align}
f(\vec{p}) & =\sum_{\nu\in S_{t}^{\oplus\eta}}\sum_{\gamma\in S_{\vec{p}_{1},\cdots,\vec{p}_{t}}}\sum_{\mu,\chi\in S_{\eta}^{\oplus t}}(-1)^{\mu\chi}\wg{\mu\nu\chi\gamma}.
\end{align}
\end{proof}
The diagonal structure factor in \ref{lem:twirling_expression_diagonal_full_proof}
contains many instances of Weingarten functions, each depending on
permissible swaps between the vectors $\vec{p}_{j}$. We now show
that for each $\eta$, there is only one unique sum.
\begin{lem}
\label{lem:diagonal_structure_factor}\label{lem:structurefactordiagonal}For
all $\eta\ge0,t=2$, $\vec{p}_{1},\vec{p}_{2}\in\mathcal{S}_{n,\eta}$,
\begin{align}
f(\vec{p}) & =g_{\eta}\left(\match{\vec{p}_{1}}{\vec{p}_{2}}\right)\Xi_{n,\eta},\quad\Xi_{n,\eta}=\sum_{\nu\in S_{t}^{\oplus\eta}}\sum_{\mu\in S_{\eta}^{\oplus t}}(-1)^{\mu}\wg{\mu\nu},\\
g_{\eta}(k) & =\left(\eta!\right)^{2}\frac{\eta+1}{\eta-k+1}.
\end{align}
\end{lem}
\begin{proof}
We make use of the fact that the Weingarten function depends only
on cycle structure of the permutation, which is uniquely determined
by its conjugacy class. Hence we have the identities such as 
\begin{equation}
\wg{abc}=\wg{cab}=\wg{c^{-1}b^{-1}a^{-1}}.
\end{equation}
From \ref{lem:twirling_expression_diagonal_full_proof}, the structure
factor 
\begin{align}
f(\vec{p}) & =\sum_{\nu\in S_{2}^{\oplus\eta}}\sum_{\gamma\in S_{\vec{p}_{1},\vec{p}_{2}}}\sum_{\mu,\chi\in S_{\eta}^{\oplus2}}(-1)^{\mu\chi}\wg{\mu\nu\chi\gamma}.
\end{align}
Let us insert an identity term $\mu\nu\chi\gamma=\mu\nu\chi\gamma\chi^{-1}\chi$
and substitute $\chi\mu\rightarrow\mu$ to obtain
\begin{align}
f(\vec{p}) & =\sum_{\chi\in S_{\eta}^{\oplus2}}\sum_{\gamma\in S_{\vec{p}_{1},\vec{p}_{2}}}\Xi_{n,\eta}\left(\chi\gamma\chi^{-1}\right),\\
\Xi_{n,\eta}\left(\tau\right) & \doteq\sum_{\mu\in S_{\eta}^{\oplus2}}\sum_{\nu\in S_{2}^{\oplus\eta}}(-1)^{\mu}\wg{\mu\nu\tau}.
\end{align}

Whereas $\gamma\in S_{\vec{p},\vec{q}}$ is a product of up to $k=\match{\vec{p}}{\vec{q}}$
$2$-cycles $\left(p_{j}q_{j}\right)$ that transposes matching elements
of $\vec{p},\vec{q}$, let the set $\Theta_{k}\ni\chi\gamma\chi^{-1}$
represents all possible transpositions $\left(p_{j}q_{k}\right)$
between up to any $k$ elements of $\vec{p},\vec{q}$. As the cycle
structure is invariant under conjugation, $\mathrm{cycles}\left(\chi\gamma\chi^{-1}\right)=\mathrm{cycles}\left(\gamma\right)$.
Let the set 
\begin{align}
\Theta_{\eta,j} & =\left\{ \chi\gamma\chi^{-1}:\chi\in S_{\eta}^{\oplus2},\gamma\in S_{2}^{\oplus\eta},\mathrm{2\text{-}cycles}\left(\gamma\right)=j\right\} ,\\
\left|\Theta_{\eta,j}\right|= & \binom{\eta}{j}\frac{(\eta)(\eta-1)\cdots(\eta-j+1)}{j!}=\binom{\eta}{j}\frac{\eta!}{(\eta-j)!}=j!\binom{\eta}{j}^{2}
\end{align}
be the distinct elements generated by any $\gamma$ with a cycle structure
of $(\underbrace{2,\cdots,2}_{j\;\text{times}},\cdots)$. Then 
\begin{equation}
\Theta_{k}=\bigcup_{j=0}^{k}\Theta_{\eta,j},\quad\left|\Theta_{k}\right|=\sum_{j=0}^{k}j!\binom{\eta}{j}^{2}.
\end{equation}
As there are fewer elements in $\left|\Theta_{k}\right|$ than there
are permutations in $\left|S_{\eta}^{\oplus2}\right|\left|S_{2}^{\oplus k}\right|$,
the map from $\left(\chi,\gamma\right)\rightarrow\tau$ is injective
with multiplicity
\begin{align}
\Theta_{\eta,k}^{-1}(\tau) & =\left\{ \left(\mu,\nu\right):\tau=\mu\nu\mu^{-1},\mu\in S_{\eta}^{\oplus2},\nu\in S_{2}^{\oplus k}\right\} ,\\
\left|\Theta_{\eta,k}^{-1}(\tau)\right| & =\left|\Theta_{\eta,k,j}^{-1}\right|=\binom{k}{j}\frac{(\eta!)^{2}}{j!\binom{\eta}{j}^{2}}=\frac{k!(\eta-j)!^{2}}{(k-j)!},\;j=2\text{-}\mathrm{cycles}(\tau)
\end{align}
The elements of $\left(\mu,\cdot\right)\in\Theta_{k}^{-1}(\tau)$
are permutations within each vector $\vec{p},\vec{q}$

This allows us to express the structure factor Weingarten sum over
distinct elements in $\Theta_{\eta}$
\begin{align}
\Xi_{n,\eta,j} & \doteq\sum_{\tau\in\Theta_{\eta,j}}\Xi_{n,\eta}\left(\tau\right),\\
\Rightarrow f(k) & =\sum_{j=0}^{k}\left|\Theta_{k,j}^{-1}\right|\Xi_{n,\eta,j}.
\end{align}
We now evaluate $\Xi_{n,\eta,j+1}$ in terms of $\Xi_{n,\eta,j}$.
Observe that the case $j=0$ corresponds to $\tau\in\text{\ensuremath{\Theta_{\eta,0}}=\ensuremath{\left\{  e\right\} } }.$
Hence
\begin{equation}
\Xi_{n,\eta,0}=\Xi_{n,\eta}\left(e\right)=\Xi_{n,\eta}.
\end{equation}
For any $j>0$, observe that any element of $\tau_{j+1}\in\Theta_{\eta,j+1}$
is the product of a transposition and an element $\tau_{j}\in\Theta_{\eta,j}$.
Let $G_{\tau_{j}}=\left\{ (jk):(jk)\notin\tau_{j}\right\} $ be the
set of transpositions between any element of $\vec{p}$ and $\vec{q}$,
excluding those contained in $\tau$. There are $\left|G_{\tau_{j}}\right|=(\eta-j)^{2}$
such transpositions. Hence
\begin{align}
\Xi_{n,\eta,j+1} & =\sum_{\tau\in\Theta_{\eta,j+1}}\Xi_{n,\eta}\left(\tau\right)=\sum_{\tau\in\Theta_{\eta,j}}\sum_{g\in G_{\tau}}\Xi_{n,\eta}\left(g\tau\right)
\end{align}

When $j=1$, consider all transpositions $g\in G_{e}$ in the sum
\begin{equation}
\Xi_{n,\eta,1}=\sum_{g\in G_{e}}\Xi_{n,\eta}\left(g\right)=\sum_{g\in G_{e}}\sum_{\mu\in S_{\eta}^{\oplus2}}\sum_{\nu\in S_{2}^{\oplus\eta}}(-1)^{\mu}\wg{\mu\nu g}.
\end{equation}
 Observe that only the $\eta$ elements $g\in S_{2}^{\oplus\eta}\cap G$
leave the sum unchanged as seen by a change of variables $\nu g\rightarrow\nu\in S_{2}^{\oplus\eta}$.
For all other transpositions in $G_{e}/S_{2}^{\oplus\eta}$, consider
the transposition $(l_{j}r_{k})$ representing swapping elements $p_{1,j}\leftrightarrow p_{2,k}$
where $j\neq k$. For every $\nu\in S_{2}^{\oplus\eta}$, with $k$
$2$-cycles, e.g. for $k=1$, $g\nu=(l_{j}r_{k})(l_{r}r_{r})$ consider
the $\nu$ with one more or less $2$-cycles where the added cycle
shares an index with $g$. E.g., $(l_{j}r_{k})(l_{r}r_{r})(l_{k}r_{k})=(l_{k}l_{j}r_{k})(l_{r}r_{r})$
or $(l_{j},r_{k})(l_{r},r_{r})(l_{j},r_{j})=(l_{j}r_{j}r_{k})(l_{r},r_{r})$.
There is always an odd permutation $\mu'\in S_{\eta}^{\oplus2}$ that
converts the $3$-cycle back into a $2$-cycle. For instance, $(r_{j}r_{k})(l_{j}r_{j}r_{k})=(l_{j}r_{k})$.
Thus the sum $\sum_{\mu\in S_{\eta}^{\oplus2}}(-1)^{\mu}\wg{\mu g\nu}=\sum_{\mu\in S_{\eta}^{\oplus2}}(-1)^{\mu}\wg{\mu\mu'g\nu}=-\sum_{\mu\in S_{\eta}^{\oplus2}}(-1)^{\mu}\wg{\mu g\nu}=0$,
and 
\begin{equation}
\Xi_{n,\eta,1}=\eta\Xi_{n,\eta,0}=\eta\Xi.
\end{equation}
For any $j>0$, a similar argument holds -- only $\eta-j$ transpositions
in $G_{\tau_{j}}$ do not cancel. All other elements $g\nu$ have
a matching $\mu'g\nu'$ with exactly the same cycle structure where
$\mu'$ is odd. Hence 
\begin{equation}
\Xi_{n,\eta,j+1}=\left(\eta-j\right)\Xi_{n,\eta,j}=\frac{\eta!}{(\eta-j)!}\Xi_{n,\eta}=j!\binom{\eta}{j}\Xi_{n,\eta}.
\end{equation}

Now substituting our result for $\Xi_{n,\eta,j}$ into the structure
factor,
\begin{equation}
f(k)=\sum_{j=0}^{k}\frac{k!(\eta-j)!^{2}}{(k-j)!}\frac{\eta!}{(\eta-j)!}\Xi_{n,\eta}=\left(\eta!\right)^{2}\sum_{j=0}^{k}\frac{k!(\eta-j)!}{\eta!(k-j)!}\Xi_{n,\eta}=\left(\eta!\right)^{2}\frac{\eta+1}{\eta-k+1}\Xi_{n,\eta}.
\end{equation}
Above, we use the fact $A(\eta,k)=\sum_{j=0}^{k}\frac{k!(\eta-j)!}{\eta!(k-j)!}=\frac{\eta+1}{\eta-k+1}$.
which may be proven by induction. Assuming $A(\eta,k)=\frac{\eta+1}{\eta-k+1}$
is true. Then $A(\eta,0)=\sum_{j=0}^{0}\frac{0!(\eta-j)!}{\eta!(0-j)!}=1$
is true. Observe that
\begin{align}
A(\eta,k) & =\sum_{j=0}^{k}\frac{k!(\eta-j)!}{\eta!(k-j)!}=\frac{\eta+1}{k+1}\sum_{j=1}^{k+1}\frac{(k+1)!(\eta+1-j)!}{(\eta+1)!(k+1-j)!}.\\
\frac{k+1}{\eta+1}\left(A(\eta,k)+\frac{\eta+1}{k+1}\right) & =A(\eta+1,k+1)
\end{align}
Now $A(\eta,k)=\frac{\eta+1}{\eta-k+1}$ is true implies that 
\begin{equation}
\frac{k+1}{\eta+1}\left(\frac{\eta+1}{\eta-k+1}+\frac{\eta+1}{k+1}\right)=\frac{\eta+2}{\eta-k+2},
\end{equation}
and $A(\eta+1,k+1)$ is true. By iterating from $A(\eta,0)$ for all
$\eta$, hence $A(\eta,k)$ is true for all $k\le\eta$.
\end{proof}
With the partially evaluated twirling operator, we may evaluate some
useful sums of Weingarten functions.
\begin{lem}
\label{lem:weingartensum} The sum of Weingarten functions 
\begin{align}
\Xi_{n,\eta} & =\sum_{\nu\in S_{2}^{\oplus\eta}}\sum_{\mu\in S_{\eta}^{\oplus2}}(-1)^{\mu}\wg{\mu\nu}=\frac{1}{\left(\eta!\right)^{2}\binom{n}{\eta}\binom{n+1}{\eta}}.
\end{align}
\end{lem}
\begin{proof}
We use the fact that the twirling operator has unit trace $\tr{}{\mathcal{T}_{2,\wedge^{\eta}\mathcal{U}_{n}}}$.
From its expression in \ref{lem:twirling_expression_diagonal_full_proof}
combined with the form for the structure factor $f(\vec{p})=g_{\eta}(\match{\vec{p}_{1}}{\vec{p}_{2}})\Xi_{n,\eta}$,
where $g_{\eta}(k)=\left(\eta!\right)^{2}\frac{\eta+1}{\eta-k+1}$
in \ref{lem:diagonal_structure_factor},
\begin{equation}
\tr{}{\mathcal{T}_{2,\wedge^{\eta}\mathcal{U}_{n}}}=\text{\ensuremath{\sum_{\vec{p}\in\mathcal{S}_{n,\eta}^{\otimes2}}}}g_{\eta}\left(\match{\vec{p}_{1}}{\vec{p}_{2}}\right)\Xi_{n,\eta}=\Xi_{n,\eta}\sum_{k=0}^{\eta}g_{\eta}\left(k\right)\text{\ensuremath{\sum_{\vec{p}_{1}\in\mathcal{S}_{n,\eta}}}}\sum_{\mathrm{Sim}\left(\vec{p}_{1},\vec{p}_{2}\right)=k}1.
\end{equation}
The combinatorial factor 
\begin{equation}
\text{\ensuremath{\sum_{\vec{p}_{1}\in\mathcal{S}_{n,\eta}}}}\sum_{\left|\vec{p}_{1}\cap\vec{p}_{2}\right|=k}1=\binom{n}{\eta}\binom{\eta}{k}\binom{n-\eta}{\eta-k}.
\end{equation}
Hence
\begin{align}
\tr{}{\mathcal{T}_{2,\wedge^{\eta}\mathcal{U}_{n}}} & =\Xi_{n,\eta}\binom{n}{\eta}\left(\eta!\right)^{2}\sum_{k=0}^{\eta}\frac{\eta+1}{\eta-k+1}\binom{\eta}{k}\binom{n-\eta}{\eta-k}=1.\label{eq:twirling_2_trace}
\end{align}
We now evaluate the sum

\begin{align}
A & =\sum_{k=0}^{\eta}\frac{\eta+1}{\eta-k+1}\binom{\eta}{k}\binom{n-\eta}{\eta-k}=\sum_{k=0}^{\eta}\frac{\eta+1}{k+1}\binom{\eta}{\eta-k}\binom{n-\eta}{k}.
\end{align}
Using the recurrence relation $\binom{n}{k}=\frac{n+1-k}{k}\binom{n}{k-1}=\frac{n+1}{k}\binom{n}{k-1}-\binom{n}{k-1}$
implies $\frac{n+1}{k}\binom{n}{k-1}=\binom{n}{k}+\binom{n}{k-1}=\binom{n+1}{k}$.
Hence the term $\frac{1}{k+1}\binom{n-\eta}{k}=\frac{1}{n-\eta+1}\binom{n-\eta+1}{k+1}$,
and
\begin{equation}
A=\frac{\eta+1}{n-\eta+1}\sum_{k=0}^{\eta}\binom{n-\eta+1}{k+1}\binom{\eta}{\eta-k}=\frac{\eta+1}{n-\eta+1}\sum_{k=0}^{\eta}\binom{n-\eta+1}{\eta-k+1}\binom{\eta}{k}.
\end{equation}
We perform the sum using the Chu-Vandermonde identity $\sum_{k=0}^{\eta}\binom{n-m}{\eta-k}\binom{m}{k}=\binom{n}{\eta}$,
to obtain
\begin{equation}
A=\frac{\eta+1}{n-\eta+1}\binom{n+1}{\eta+1}=\binom{n+1}{\eta}.
\end{equation}
Substituting $A$ back into \ref{eq:twirling_2_trace}, we obtain
$\Xi_{n,\eta}=\frac{1}{(\eta!)^{2}}\binom{n+1}{\eta}^{-1}\binom{n}{\eta}^{-1}$.
\end{proof}

\section{\label{sec:Hypergeometric-sums}Hypergeometric sums}

In this section we evaluate the hypergeometric sums associated with
$\tr{}{\tilde{n}_{d}^{2}}$ in \ref{eq:trace_squared_sum} of \ref{thm:inverse_measurement_channel}
and the estimation matrix entries $E_{\eta,k,s}$ in \ref{eq:estimation_matrix_entries}
of \ref{thm:single_shot_estimate_efficient}.
\begin{lem}
\label{lem:trace_symmetric_difference_squared}For all integers $n\ge\eta\ge0$,
$d\in[0,\min(n,n-\eta)]$,
\begin{align}
\tr{\eta}{\tilde{n}_{d}^{2}} & =\frac{\sum_{s=0}^{\min(\eta,n-\eta)}\binom{\eta}{\eta-s}\binom{n-\eta}{s}\left(\sum_{j=\max(0,s+d-\eta)}^{\min(d,s)}(-1)^{j}(\eta-d+j)!(n-\eta-j)!\binom{s}{j}\binom{\eta-s}{d-j}\right)^{2}}{(\eta-d)!^{2}(n-\eta-d)!^{2}}\nonumber \\
 & =\frac{\eta!}{d!}\frac{(n-d+1)!(n-\eta)!}{(n-2d+1)(n-\eta-d)!^{2}(\eta-d)!^{2}}.
\end{align}
\end{lem}
\begin{proof}
From \ref{eq:symmetric_difference_elementary_polynomials}, $\tilde{n}_{d}=\sum_{j=0}^{d}(-1)^{j}\frac{(\eta-d+j)!}{(\eta-d)!}\frac{(n-\eta-j)!}{(n-\eta-d)!}e_{d-j}(\hat{n}_{1},\cdots,\hat{n}_{\eta})e_{j}(\hat{n}_{\eta+1},\cdots,\hat{n}_{n})$.
Hence, the trace of its square
\begin{align}
\tr{\eta}{\tilde{n}_{d}^{2}} & =\frac{\tr{\eta}{\left(\sum_{j=0}^{d}(-1)^{j}(\eta-d+j)!(n-\eta-j)!e_{d-j}(\hat{n}_{1},\cdots,\hat{n}_{\eta})e_{j}(\hat{n}_{\eta+1},\cdots,\hat{n}_{n})\right)^{2}}}{(\eta-d)!^{2}(n-\eta-d)!^{2}}\nonumber \\
 & =\sum_{j,k=0}^{d}\frac{(-1)^{j+k}(\eta-d+j)!(n-\eta-j)!(\eta-d+k)!(n-\eta-k)!}{(\eta-d)!^{2}(n-\eta-d)!^{2}}t_{\eta,d,j,k},\label{eq:trace_of_symmetric_difference_squared}\\
t_{\eta,d,j,k} & =\tr{\eta}{e_{d-j}(\hat{n}_{1},\cdots,\hat{n}_{\eta})e_{j}(\hat{n}_{\eta+1},\cdots,\hat{n}_{n})e_{d-k}(\hat{n}_{1},\cdots,\hat{n}_{\eta})e_{k}(\hat{n}_{\eta+1},\cdots,\hat{n}_{n})}.
\end{align}
Let $\ket{\eta,s}$ be any state with $s$ fermions supported on the
number operators $\hat{n}_{\eta+1},\cdots,\hat{n}_{n}$, that is $\sum_{j\in[n]\backslash[\eta]}\hat{n}_{j}\ket{\eta,s}=s\ket{\eta,s}$
and $\sum_{j\in[\eta]}\hat{n}_{j}\ket{\eta,s}=\eta-s\ket{\eta,s}$.
Note that there are $\binom{\eta}{s}\binom{n-\eta}{s}$ such states.
Observe that
\begin{equation}
e_{j}(\hat{n}_{1},\cdots,\hat{n}_{\eta})\ket{\eta,s}=\binom{\eta-s}{j}\ket{\eta,s},\quad e_{k}(\hat{n}_{\eta+1},\cdots,\hat{n}_{n})\ket{\eta,s}=\binom{s}{k}\ket{\eta,s}.
\end{equation}
Hence the trace 
\begin{align}
t_{\eta,d,j,k} & =\sum_{s=0}^{\min(\eta,n-\eta)}\binom{\eta}{s}\binom{n-\eta}{s}\bra{\eta,s}e_{d-j}(\hat{n}_{1},\cdots,\hat{n}_{\eta})e_{j}(\hat{n}_{\eta+1},\cdots,\hat{n}_{n})\cdots\ket{\eta,s}\nonumber \\
 & =\sum_{s=0}^{\min(\eta,n-\eta)}\binom{\eta}{s}\binom{n-\eta}{s}\binom{s}{j}\binom{s}{k}\binom{\eta-s}{d-j}\binom{\eta-s}{d-k}.
\end{align}
Substituting into \ref{eq:trace_of_symmetric_difference_squared},
and noting that the summand is zero when $d-j>\eta-s$ due to the
term $\binom{\eta-s}{d-j}$ and similarly for $\binom{\eta-s}{d-k}$,
we obtain the sum 
\begin{align}
\tr{\eta}{\tilde{n}_{d}^{2}} & =\frac{\sum_{s=0}^{\min(\eta,n-\eta)}\sum_{j,k=\max(0,s+d-\eta)}^{\min(d,s)}F_{n,\eta,d}(s,j,k)}{(\eta-d)!^{2}(n-\eta-d)!^{2}},
\end{align}
where the summand
\begin{align}
F_{n,\eta,d}(s,j,k) & =(-1)^{j+k}(\eta-d+j)!(n-\eta-j)!(\eta-d+k)!(n-\eta-k)!\nonumber \\
 & \qquad\times\binom{\eta}{s}\binom{n-\eta}{s}\binom{s}{j}\binom{s}{k}\binom{\eta-s}{d-j}\binom{\eta-s}{d-k}.
\end{align}
By recognizing the double sum over $j,k$ as the square of a sum,
\begin{align}
\tr{\eta}{\tilde{n}_{d}^{2}} & =\frac{\eta!(n-\eta)!(\eta-s)!\sum_{s=0}^{\min(\eta,n-\eta)}\left(\sum_{j=\max(0,s+d-\eta)}^{\min(d,s)}f_{n,\eta,d}(s,j)\right)^{2}}{(n-\eta-s)!(\eta-d)!^{2}(n-\eta-d)!^{2}},\\
f_{n,\eta,d}(s,j) & =(-1)^{j}\frac{(\eta-d+j)!(n-\eta-j)!}{j!(s-j)!(d-j)!(\eta-s-d+j)!}.
\end{align}
The above sum holds for all non-negative integers satisfying $\eta\le n,d\le\min(\eta,n-\eta)$.
This proves the first equality in \ref{eq:trace_of_symmetric_difference_squared}. 

We find it convenient to define
\begin{align}
F_{n,\eta,d}(j,k) & =\sum_{s=0}^{\min(\eta,n-\eta)}F_{n,\eta,d}(s,j,k),\\
F_{n,\eta,d}(k) & =\sum_{j=\max(0,s+d-\eta)}^{\min(d,s)}F_{n,\eta,d}(j,k),\\
F_{n,\eta,d} & =\sum_{k=\max(0,s+d-\eta)}^{\min(d,s)}F_{n,\eta,d}(k).
\end{align}
For example, the sum simplifies when $d=0$ as only the summand is
non-zero only when $j=k=0$. Hence
\begin{align}
F_{n,\eta,0} & =\sum_{s=0}^{\min(\eta,n-\eta)}F_{n,\eta,d}(s,0,0)=(\eta)!^{2}(n-\eta)!^{2}\sum_{s}\binom{\eta}{s}\binom{n-\eta}{s}\nonumber \\
 & =(\eta)!^{2}(n-\eta)!\binom{n}{n-\eta}=(\eta)!^{2}(n-\eta)!n!,
\end{align}
where we apply the Chu-Vandermonde identity $\sum_{j=0}^{k}\binom{m}{j}\binom{n-m}{k-j}=\binom{n}{k}$
for any complex $m,n$ and integer $k\ge0$. 

Our proof strategy is to first find a linear recurrence in $d$ satisfied
by $F_{2\eta,\eta,d}$ with $F_{2\eta,\eta,0}$ as the initial condition,
and second to find a linear recurrence in $n$ satisfied by $F_{n,\eta,d}$
with $F_{2\eta,\eta,d}$ as the initial condition. A more direct proof
would find a linear recurrence in $d$ satisfied by $F_{n,\eta,d}$,
but we were unable to do so in reasonable time. Let the shift operator
on the variable $x$ be $S_{x}f(x)=f(x+1)$. From the definition of
the summand, it is straightforward to verify that it is annihilated
like $Q_{i}(s,j,k)F_{n,\eta,d}(s,j,k)=0$ by the difference operators
\begin{align}
Q(s,j,k)=\{ & (j-s-1)(-k+s+1)(s-\eta)S_{s}-(n-\eta-s)(d-\eta-j+s)(d-\eta-k+s),\nonumber \\
 & (k+1)(\eta+k-n)(-d+\eta+k-s+1)S_{k}-(d-k)(k-s)(d-\eta-k-1),\nonumber \\
 & (d-j+1)(d-k+1)(d-\eta-j)(d-\eta-k)S_{d}-(d-\eta-j+s)(d-\eta-k+s),\nonumber \\
 & (j+1)(\eta+j-n)(\eta-d+j-s+1)S_{j}-(d-j)(j-s)(d-\eta-j-1),\nonumber \\
 & (n-\eta-s+1)S_{n}-(n-\eta+1)(\eta+j-n-1)(\eta+k-n-1)\}.\label{eq:Wilf-Zeilberger_0}
\end{align}

The amazing method of creative telescoping by Wilf and Zeilberger
\citep{petkovvsek1996b} guarantees the existence of $s$-free operators
$Q_{i}(j,k)$ that do not depend on $s$ and some certificate $R_{i}(s,j,k)$
that annihilate this proper hypergeometric summand according to
\begin{equation}
\left[Q_{i}(j,k)+(S_{s}-1)R_{i}(s,j,k)\right]F_{n,\eta,d}(s,j,k)=0.\label{eq:Wilf-Zeilberger_recurrence}
\end{equation}
The $s$-free property of $Q$ allows us to perform a sum over $s\in[s_{0},s_{1}]$.
\begin{align}
\sum_{s=s_{0}}^{s_{1}}Q_{i}(j,k)F_{n,\eta,d}(s,j,k) & =\sum_{s=s_{0}}^{s_{1}}(S_{s}-1)R_{i}(s,j,k)F_{n,\eta,d}(s,j,k),\label{eq:Wilf-Zeilberger_1}\\
\Rightarrow Q_{i}(j,k)\sum_{s=s_{0}}^{s_{1}}F_{n,\eta,d}(s,j,k) & =R_{i}(s_{1},j,k)F_{n,\eta,d}(s_{1},j,k)-R_{i}(s_{0},j,k)F_{n,\eta,d}(s_{0},j,k).
\end{align}
Hence $Q$ defines the linear recurrence that the sum $\sum_{s=s_{0}}^{s_{1}}F_{n,\eta,d}(s,j,k)$
satisfies. A key insight in evaluating the right-hand side is that
the summand has compact support. Using the Euler's reflection formula
$(-z)!=\frac{\pi}{(z-1)!\sin(\pi z)}$ and canceling poles when $z$
approaches an integer, observe that
\begin{equation}
F_{n,\eta,d}(s,j,k)=\frac{(-1)^{j+k}\eta!(n-\eta)!(\eta-s)!(\eta-d+j)!(\eta-d+k)!(n-\eta-j)!(n-\eta-k)!}{j!k!(d-j)!(d-k)!(s-j)!(s-k)!(n-\eta-s)!(\eta-d+j-s)!(\eta-d+k-s)!},
\end{equation}
 is zero whenever the number of factorials with negative arguments
in the denominator is greater than that in the numerator. Hence we
may change the summation limits to, for instance, 
\begin{equation}
F_{n,\eta,d}(j,k)=\sum_{s=0}^{\min(\eta,n-\eta)}F_{n,\eta,d}(s,j,k)=\sum_{s=-1}^{\min(\eta,n-\eta)+1}F_{n,\eta,d}(s,j,k).
\end{equation}
This allows the telescoping sum $\sum_{s=-1}^{\min(\eta,n-\eta)}(S_{s}-1)R_{i}(s,j,k)F_{n,\eta,d}(s,j,k)=0$
due to the natural boundaries of $F_{n,\eta,d}(s,j,k)$. Thus $Q_{i}(j,k)$
defines a recurrence satisfied by $F_{n,\eta,d}(j,k)$. By recursing
the procedure and finding $(s,j)$-free $Q_{i}(k)$ and then $(s,j,k)$-free
$Q_{i}$, we then obtain the desired recurrence for $F_{n,\eta,d}$
as follows

\begin{align}
\sum_{k}\left[Q_{i}(k)+(S_{j}-1)R_{i}(j,k)\right]F_{n,\eta,d}(j,k) & =Q_{i}(j,k)F_{n,\eta,d}(j)=0,\\
\sum_{j}\left[Q_{i}+(S_{k}-1)R_{i}(k)\right]F_{n,\eta,d}(j) & =Q_{i}(j,k)F_{n,\eta,d}=0,
\end{align}

At each iteration $x$ of the recursion, the operators $R(x,\cdots)$
are also called certificates as the correctness of $Q(\cdots)$ may
be readily verified by reducing each equation $\left[Q_{i}(\cdots)+(S_{x}-1)R_{i}(x,\cdots)\right]$
with respect to the annihilators $Q_{i}(x,\cdots)$ of the preceding
iteration. For instance, reducing \ref{eq:Wilf-Zeilberger_recurrence}
with respect to the annihilators \ref{eq:Wilf-Zeilberger_0} may be
done by hand, through verification of later iterations should be done
by computer. 

We now state the recurrences $\left[\tilde{Q}_{i}(\cdots)+(S_{x}-1)\tilde{R}_{i}(s,\cdots)\right]F_{2\eta,\eta,d}(x,\cdots)=0$
with respect to $d,s,j,k$. These were all computed in Mathematica
by the HolonomicFunctions package \citep{kouschan2014}. The $s$-free
operators and certificates are
\begin{align}
\tilde{Q}(j,k)=\{ & (d-k)(d-\eta-k-1)(d-\eta-j+k)-(k+1)(k-\eta)(\eta-d-j+k+1)S_{k},\nonumber \\
 & (d-j)(d-\eta-j-1)(d-\eta+j-k)-(j+1)(j-\eta)(\eta-d+j-k+1)S_{j},\nonumber \\
 & (d-\eta)(d-\eta+j-k)(d-\eta-j+k)\nonumber \\
 & \quad-2(2d-2\eta+1)(d-j+1)(d-k+1)(d-\eta-j)(d-\eta-k)S_{d}\},\\
\tilde{R}(s,j,k)=\{ & \frac{(d-k)(s-j)(k-s)(d-\eta-k-1)}{d-\eta-k+s-1},\frac{(d-j)(j-s)(s-k)(d-\eta-j-1)}{d-\eta-j+s-1},\nonumber \\
 & (j-s)(s-k)(3\eta-3d+j+k-2s)\}.
\end{align}
The $(s,j)$-free operators and certificates are
\begin{align}
\tilde{Q}(k)=\big\{ & (2+k)(1+k-\eta)S_{k}^{2}+\left(-2+2(-2+d)k-2k^{2}+d(3-d+\eta)\right)S_{k}\nonumber \\
 & \quad+(-d+k)(1-d+k+\eta),\nonumber \\
 & 2(1+d)(-1-d+k)(-1+d-2\eta)(1+2d-2\eta)(-d+k+\eta)S_{d}\nonumber \\
 & \quad+(1+k)(k-\eta)(-1-3d+2k+2\eta)S_{k}-\eta\left(8d^{2}-9dk+d+4k^{2}+k-1\right)\nonumber \\
 & \quad+(2d-k)\left(-3dk+d(2d-1)+2k^{2}+k-1\right)+\eta^{2}(4d-2k+2)\big\},
\end{align}
\begin{align}
\tilde{R}(k)=\big\{ & \frac{j(-d+k)(-1+2d-2\eta)(-1+j-\eta)(d-j+k-\eta)(1-d+k+\eta)}{(1+k)(k-\eta)(1-d-j+k+\eta)(2-d-j+k+\eta)},\nonumber \\
 & \left((-1-d+j)(-1+d+j-k-\eta)(-d+j+\eta)\right)^{-1}j(-1+j-\eta)(-d+j-k+\eta)\nonumber \\
 & \quad\times\big(\left(-2+d(5+17d)+j-5dj+2(-2-7d+2j)k+2k^{2}\right)\eta\nonumber \\
 & \quad+(-2d+k)\left(-1+d+4d^{2}-2dj+(-1-3d+2j)k\right)\nonumber \\
 & \quad+(-3-11d+2j+4k)\eta^{2}+2\eta^{3}\big)\big\}.
\end{align}
The $(s,j,k)$-free operators and certificates are
\begin{align}
\tilde{Q}=\{ & (2d-2\eta-1)+(d+1)(d-2\eta-1)(2d-2\eta+1)S_{d}\},\\
\tilde{R}(k)=\big\{ & \frac{(1+3d-2k-2\eta)\left((1+k)(k-\eta)S_{k}+\left(-2dk+2k^{2}+d(-1+d-\eta)\right)\right)}{2(-1-d+k)(-d+k+\eta)}\}.
\end{align}

Hence, $\tilde{Q}$ defines the recurrence satisfied by
\begin{align}
F_{2\eta,\eta,d} & =\frac{(\eta-d+3/2)}{d(2\eta-d+2)(\eta-d+1/2)}F_{2\eta,\eta,d-1}\nonumber \\
 & =\frac{(\eta+1/2)!}{d!(\eta-d+1/2)!}\frac{(2\eta-d+1)!}{(2\eta+1)!}\frac{(\eta-d-1/2)!}{(\eta-1/2)!}F_{2\eta,\eta,0}\nonumber \\
 & =\frac{(2\eta-d+1)!}{d!(2\eta-2d+1)(2\eta)!}F_{2\eta,\eta,0}\nonumber \\
 & =\frac{\eta!^{2}(2\eta-d+1)!}{d!(2\eta-2d+1)}.
\end{align}
We now state the recurrences $\left[Q_{i}(\cdots)+(S_{x}-1)R_{i}(s,\cdots)\right]F_{n,\eta,d}(x,\cdots)=0$
with respect to $n,s,j,k$ and their certificates. The $s$-free operators
and certificates are
\begin{align}
 & Q(j,k)=\{\nonumber \\
 & (\eta-n-1)(d-\eta-k-1)\left(d^{2}-dj-\eta(3d+k-2(n+1))+dn+jk-k^{2}+kn-n^{2}-n\right)\nonumber \\
 & \quad+(k+1)(n-\eta+1)(\eta+k-n)(\eta-d-j+k+1)S_{k}+(n-2\eta)(\eta-d+k+1)S_{n},\nonumber \\
 & (\eta-n-1)(d-\eta-j-1)\left(d^{2}+d(-3\eta-k+n)-j^{2}+j(-\eta+k+n)-(n+1)(n-2\eta)\right)\nonumber \\
 & \quad+(j+1)(n-\eta+1)(\eta+j-n)(\eta-d+j-k+1)S_{j}+(n-2\eta)(\eta-d+j+1)S_{n},\nonumber \\
 & (n-\eta+2)\left(d^{2}-d(4\eta+j+k-2n-2)-\eta^{2}-\eta(j+k-4n-6)+(n+1)(j+k-2n-3)\right)S_{n}\nonumber \\
 & \quad-(2d-n-1)(n-\eta+1)(n-\eta+2)(\eta+j-n-1)(\eta+k-n-1)+(n-2\eta+1)S_{n}^{2}\},
\end{align}
\begin{align}
R(s,j,k)=\{ & \frac{(j-s)(k-s)(-\eta+n+1)(s-1-\eta)(-d+\eta+k+1)(d-2\eta-k+n)}{(-\eta+n-s+1)(d-\eta-k+s-1)},\nonumber \\
 & \frac{(j-s)(s-k)(-\eta+n+1)(s-1-\eta)(-d+\eta+j+1)(-d+2\eta+j-n)}{(\eta-n+s-1)(-d+\eta+j-s+1)},\nonumber \\
\text{} & \frac{(j-s)(k-s)(\eta-n-2)(\eta-n-1)(s-1-\eta)(\eta+j-n-1)(\eta+k-n-1)}{(\eta-n+s-2)(\eta-n+s-1)}\}.
\end{align}
The $(s,j)$-free operators and certificates are

\begin{align}
Q(k)=\{ & (-\eta+n+1)\left(d^{2}+d(\eta-2n-3)+k^{2}+k(\eta-n)+(n+1)(-\eta+n+1)\right)\nonumber \\
 & \quad+(d+\eta-n-1)S_{n}-((k+1)(-\eta+n+1)(\eta+k-n))S_{k},\nonumber \\
 & (-\eta+n+2)\left(d^{2}+d(3\eta+2k-4n-7)+(n+2)(-2\eta-k+2n+3)\right)S_{n}\nonumber \\
 & \quad+(d-n-2)(2d-n-1)(n-\eta+1)(n-\eta+2)(\eta+k-n-1)\nonumber \\
 & \quad+(d+\eta-n-2)S_{n}^{2}\},
\end{align}
\begin{align}
R(j,k)=\{ & \frac{j(2d-n-1)(-\eta+n+1)(\eta+j-n-1)}{d-\eta+j-k-1}+\frac{(j(d+\eta+j-k-n-1))}{d-\eta+j-k-1}S_{n},\nonumber \\
 & (n-2\eta+1)^{-1}\big[j(n-2d+1)(-\eta+n+1)(-\eta+n+2)(\eta+j-n-1)(-\eta-k+n+1)\nonumber \\
 & \quad+\left(d^{2}-d(2\eta+j+k-n-1)+\eta^{2}-\eta(k-j+2n+2)-(n+1)(j-n-1)\right)\big]\nonumber \\
 & \quad\times j(\eta-n-2)S_{n}\big\}.
\end{align}
The $(s,j,k)$-free operators and certificates are

\begin{align}
Q & =\{(d-n-2)(2d-n-1)(-\eta+n+1)+(2d-n-2)S_{n}\},\quad R(k)=kS_{n}.
\end{align}
Hence, $Q$ defines the recurrence satisfied by 
\begin{align}
F_{n,\eta,d} & =\frac{(n-d+1)(n-2d)(n-\eta)}{(n-2d+1)}F_{n-1,\eta,d}\nonumber \\
 & =\frac{(n-d+1)!}{(2\eta-d+1)!}\frac{(n-2d)!}{(2\eta-2d)!}\frac{(n-\eta)!}{\eta!}\frac{(2\eta-2d+1)!}{(n-2d+1)!}F_{2\eta,\eta,d}\nonumber \\
 & =\frac{(n-d+1)!}{(2\eta-d+1)!}\frac{(n-\eta)!}{\eta!}\frac{(2\eta-2d+1)}{(n-2d+1)}F_{2\eta,\eta,d}\nonumber \\
 & =\frac{(n-d+1)!}{(2\eta-d+1)!}\frac{(n-\eta)!}{\eta!}\frac{(2\eta-2d+1)}{(n-2d+1)}\frac{\eta!^{2}(2\eta-d+1)!}{d!(2\eta-2d+1)}\nonumber \\
 & =\frac{\eta!(n-\eta)!}{d!}\frac{(n-d+1)!}{(n-2d+1)}.
\end{align}
We complete the proof by dividing $\tr{\eta}{\tilde{n}_{d}^{2}}=\frac{\sum_{s,j,k}F_{n,\eta,d,s,j,k}}{(\eta-d)!^{2}(n-\eta-d)!^{2}}.$
\end{proof}

We now apply the same technique to evaluate the estimation matrix
entries $E_{\eta,k,s}$ from \ref{eq:estimation_matrix_entries}.
In the following, the sum $t_{n,\eta,k,s}=E_{\eta,k,s}$ following
a change of variables $d'\rightarrow d-d'$.
\begin{lem}
\label{lem:estimation_matrix_values}The sum 
\begin{equation}
t_{n,\eta,k,s}=\sum_{d=0}^{\eta}\sum_{d'=0}^{d}\sum_{x''=0}^{k-s}\sum_{y''=0}^{s}F_{n,\eta,k,s}(d,d',x'',y'')=(-1)^{s}\binom{k}{s}^{-1}\binom{\eta-k+s}{s}\binom{n-\eta+k-s}{k-s},
\end{equation}
where the summand
\begin{align}
F_{n,\eta,k,s}(d,d',x'',y'') & =a_{d}\binom{n+1}{d}(-1)^{d-d'}\frac{(\eta-d')!}{(\eta-d)!}\frac{(n-\eta-d+d')!}{(n-\eta-d)!}\binom{k-s}{x''}\binom{\eta-k+s}{d'-x''}\binom{s}{y''}\nonumber \\
 & \qquad\times\binom{n-\eta-s}{d-d'-y''}\binom{n-\left(d+k-x''-y''\right)}{n-\eta}.
\end{align}
\end{lem}
\begin{proof}
As $F_{n,\eta,k,s}(d,d',x'',y'')$ is zero outside the domain of summation
in $\sum_{d=0}^{\eta}\sum_{d'=0}^{d}\sum_{x''=0}^{k-s}\sum_{y''=0}^{s}\cdots$,
we may replace the summation limits with a sum over all integers
\begin{equation}
t_{n,\eta,k,s}\doteq\sum_{d}\sum_{d'}\sum_{x''}\sum_{y''}F_{s}(d,d',x'',y'').
\end{equation}
We find it convenient to define
\begin{equation}
t_{n,\eta,k,s}=\sum_{d,c,a,b'}F_{n,\eta,k,s}(d,c,a,b).
\end{equation}
The sum simplifies when $k=s=0$ to 
\begin{equation}
t_{n,\eta,0,0}=t_{n,\eta,0,0}(0,0,0,0)=1.
\end{equation}
Our proof strategy is to first find a linear recurrence in $k$ satisfied
by $t_{n,\eta,k,0}$ with $t_{n,\eta,0,0}$ as the initial condition,
and second to find a linear recurrence in $s$ satisfied by $t_{n,\eta,k,s}$
with $t_{n,\eta,k,0}$ as the initial condition. As with the proof
of \ref{lem:trace_symmetric_difference_squared}, we exhibit a sequence
of $y''$-free, $x''$-free, $d'$-free, and finally $d$-free operators
$Q$ that annihilate the summand $QF_{n,\eta,k,s}(d,d',x'',y'')$
together with certificates $R$ that verify their correctness. 

When $s=0$, $t_{n,\eta,k,0}=\sum_{d,d',x''}F_{n,\eta,k,s}(d,d',x'',0)$
is a triple sum. We now state the recurrences $\left[\tilde{Q}_{i}(\cdots)+(S_{x}-1)\tilde{R}_{i}(x,\cdots)\right]F_{n,\eta,k,s}(x,\cdots,0)=0$
with respect to $k,d,d',x''$ and their certificates. The $d'$-free
operators and certificates are
\begin{align}
Q(d,a)=\big\{ & (-1+d-k)(a-d-k+n)S_{k}+(1+k)(a-d-k+\eta),\nonumber \\
 & (1+a)(1+a-d-k+\eta)S_{a}-(a-d)(1+a-d-k+n),\nonumber \\
 & \quad-(-1+a-d)(1+d)(-1+2d-n)(a-d-k+n)(d-\eta)S_{d}\nonumber \\
 & \quad+(d-k)(1+2d-n)(1-d+n)^{2}(a-d-k+\eta)\big\},\\
R(d,c,a)=\big\{ & \frac{(a-c)(1+k)(-1+c-\eta)(a-d-k+\eta)}{(-1+a-k)(k-\eta)},\nonumber \\
 & \frac{(a-c)(1+a-d-k+n)(-1+c-\eta)}{1+a-c-k+\eta},\nonumber \\
 & \frac{(a-c)(1+2d-n)(1-d+n)^{2}(-1+c-\eta)(-a+d+k-\eta)}{-1+c-d}\big\}.
\end{align}
The $(d',x'')$-free operators and certificates are
\begin{align}
Q(d)=\big\{ & (-1+d-k)(d+k-n)S_{k}+(1+k)(k-\eta),\nonumber \\
 & (1+d)^{2}(1-2d+n)(d+k-n)(d-\eta)S_{d}+(d-k)(1+2d-n)(1-d+n)^{2}(d-n+\eta)\Big\},\\
R(d,a)=\big\{ & \frac{a(1+k)(a-d-k+\eta)}{-a+d+k-n},\nonumber \\
 & \frac{a(d-k)(1+2d-n)(1-d+n)^{2}(-1+a-2d-k+n)(a-d-k+\eta)}{(-1+a-d)(a-d-k+n)}\big\}.
\end{align}
The $(d,d',x'')$-free operators and certificates are
\begin{align}
Q=\big\{ & (-1-k)S_{k}+(1+k+n-\eta)\big\},\\
R(d)=\big\{ & \frac{d^{2}(-1+d-\eta)}{(-1+d-k)(-1+2d-n)}\big\}.
\end{align}
Hence, we solve the first order recurrence defined by $Qt_{n,\eta,k,0}$
to obtain 
\begin{align}
t_{n,\eta,k,0} & =\frac{n-\eta+k}{k}t_{n,\eta,k-1,0}=\frac{(n-\eta+k)!}{k!(n-\eta)}t_{n,\eta,0,0}=\frac{(n-\eta+k)!}{k!(n-\eta)}.
\end{align}

For the case $s>0$, we now state the recurrences $\left[\tilde{Q}_{i}(\cdots)+(S_{x}-1)\tilde{R}_{i}(x,\cdots)\right]F_{n,\eta,k,s}(x,\cdots,0)=0$
with respect to $s,d,d',x'',y''$ and their certificates. The $y''$-free
operators and certificates are
\begin{align}
Q(d,c,a)=\big\{ & (a-d-k+n)(1-a+k-s)(k-s-\eta)S_{k}\nonumber \\
 & \quad+(1+k-s)(a-c-k+\eta)(a-d-k+s+\eta),\nonumber \\
 & (k-s)(-1-a+d+k-s-\eta)S_{s}+(a-k+s)(-1+k-s-\eta),\nonumber \\
 & (-1+c-d)(1+d)(-1+2d-n)(a-d-k+n)S_{d}\nonumber \\
 & \quad-(1+2d-n)(1-d+n)^{2}(a-d-k+s+\eta),\nonumber \\
 & (-1+a-c)(c-\eta)S_{c}+(c-d)(-a+c+k-\eta),\nonumber \\
 & (1+a)(1+a-c-k+\eta)(1+a-d-k+s,+\eta)S_{a}\nonumber \\
 & \quad-(a-c)(1+a-d-k+n)(a-k+s)\big\},\\
R(d,c,a,b)=\big\{ & -\frac{b(k-s+1)(a+b-d+\eta-k)(b+c-d-\eta+n-s)}{a+b-d-k+n},\nonumber \\
 & \frac{b(a-k+s)(\eta-k+s+1)(a+b-d+\eta-k)(b+c-d-\eta+n-s)}{(b-s-1)(\eta-n+s)(a-c+\eta-k+s+1)},\nonumber \\
 & \left((b+c-d-1)(c-d-\eta+n)(a+b-d-k+n)\right)^{-1}b(2d-n+1)(-d+n+1)^{2}\nonumber \\
 & \quad\times(a+b-d+\eta-k)(a+b+c-2d-k+n-1)(-b-c+d+\eta-n+s),\nonumber \\
 & b(a+b-d+\eta-k),\nonumber \\
 & \frac{b(c-a)(a-k+s)(b+c-d-\eta+n-s)}{a-c+\eta-k+s+1}\big\}.
\end{align}
The $(d',y'')$-free operators and certificates are
\begin{align}
Q(d,a)=\big\{ & (1-d+k)(a-d-k+n)(1-a+k-s)(k-s-\eta)S_{k}\nonumber \\
 & \quad+(a-1-k)(1+k-s)(k-\eta)(a-d-k+s+\eta),\nonumber \\
 & (-1+a-d)(1+d)(-1+2d-n)(a-d-k+n)(d-\eta)S_{d}\nonumber \\
 & \quad-(d-k)(1+2d-n)(1-d+n)^{2}(a-d-k+s+\eta),\nonumber \\
 & (1+a)(k-a)(1+a-d-k+s+\eta)S_{a}+(a-d)(1+a-d-k+n)(a-k+s),\nonumber \\
 & (s-k)(-1-a+d+k-s-\eta)S_{s}-(a-k+s)(-1+k-s-\eta)\big\},\\
R(d,c,a)=\big\{ & (a-c)(1+k-s)(-1+c-\eta)(a-d-k+s+\eta)\nonumber \\
 & \frac{(a-c)(1+2d-n)(1-d+n)^{2}(-1+c-\eta)(a-d-k+s+\eta)}{-1+c-d},\nonumber \\
 & \frac{(a-c)(1+a-d-k+n)(a-k+s)(1-c+\eta)}{1+a-c-k+\eta},0\big\}.
\end{align}
The $(d',x'',y'')$-free operators and certificates are
\begin{align}
Q(d)=\big\{ & -(-1+d-n)(2d-n)(1+2d-n)(k-s)(n-s-\eta)S_{s}\nonumber \\
 & \quad+(1+d)^{2}(-1+2d-n)(d+k-n)(d-\eta)S_{d}\nonumber \\
 & \quad+\big(d^{3}+d^{2}(k-1-n-2s-\eta)+k(n+ns-\eta)-ns(1+s+\eta)\nonumber \\
 & \quad+d(2s+\eta+(n+2s)(s+\eta)-k(1+2s+\eta))\big)(-1+d-n)(1+2d-n),\nonumber \\
 & (1+d)(2+d)^{2}(1-2d+n)(2d-n)(1+d+k-n)(1+d-\eta)S_{d}^{2}\nonumber \\
 & \quad+\left(n(n(-1+k-s)+2s)+d^{2}(2k+n-4s-2\eta)-d(-1+n)(2k+n-4s-2\eta)+2n\eta-k(2+n)\eta\right)\nonumber \\
 & \quad\times(1+d)(n-d)(-1+2d-n)(3+2d-n)S_{d}\nonumber \\
 & \quad+(d-k)(d-n)(2+2d-n)(3+2d-n)(1-d+n)^{2}(d-n+\eta)\big\},\\
R(d,c)=\big\{ & \left((-1+a-d)(a-d-k+n)\right)^{-1}a(-1+a-k)(-1+d-n)(1+2d-n)\big[-d^{3}+k^{2}-kn+k^{2}n+n^{2}\nonumber \\
 & \quad-kn^{2}-ks-kns+n^{2}s-\left(k+kn-n^{2}\right)\eta+d^{2}(1-2k+3s+3\eta)+a\big(d^{2}-k(1+n)\nonumber \\
 & \quad+d(-1+k-2s-2\eta)+n(1+s+\eta)\big)+d\left(-k^{2}+s+\eta-3n(1+s+\eta)+k(2+3n+s+\eta)\right)\big],\nonumber \\
 & \left((-2+a-d)(-1+a-d)(-1+a-d-k+n)(a-d-k+n)(d-\eta)\right)^{-1}a(-1+a-k)(d-k)(d-n)\nonumber \\
 & \quad\times(3+2d-n)(1-d+n)^{2}(a-d-k+s+\eta)\big[4d^{4}-2(-2+a)(-1+a-k)\eta\nonumber \\
 & \quad-d\big(a^{2}(-2+n)+3k^{2}n+(-1+n)(4+3n-8s+5ns)-k(4+n(-8+5n+3s))\nonumber \\
 & \quad+a(6+k(2-4n)-4s+n(-4+n+4s))\big)-2d\left(9+a^{2}-a(7+k-4n)+k(3-2n)+n(-11+3n)\right)\eta\nonumber \\
 & \quad+n^{3}(1-k+s+\eta)-2d^{3}(-6+3a-3k+2n+3s+4\eta)+(-2+a)n(-2s+(-5+a-k)\eta)\nonumber \\
 & \quad+n^{2}((3-a+k)(-1+k-s)-(5-2a+k)\eta)+d^{2}\big(2a^{2}+2k^{2}-k(-10+9n+2s+4\eta)\nonumber \\
 & \quad+a(-12-4k+5n+4s+8\eta)-2(-6+7s+11\eta)+n(-5+n+9s+12\eta)\big)\big]\big\}.
\end{align}
The $(d,d',x'',y'')$-free operators and certificates are
\begin{align}
Q=\big\{ & -(n-\eta+k-s)S_{s}-(\eta+s-k+1)\big\},\\
R(d)=\big\{ & -\frac{d(1+d)^{2}(d+k-n)(d-\eta)(-1+d+k-s-\eta)}{(-1+d-n)(2d-n)(1+2d-n)(k-s)(n-s-\eta)}S_{d}\\
 & -\left((-1+2d-n)(2d-n)(k-s)(n-s-\eta)\right)^{-1}d\nonumber \\
 & \times\big(d^{4}+d^{3}(-2+2k-n-3s-2\ \eta)+(1+n)(1-k+s+\eta)(ns+k(-n+\eta))\nonumber \\
 & +d^{2}\left(1+k^{2}+n+5s+3\eta+(s+\eta)(2n+4s+\eta)-k(3+n+5s+3\eta)\right)\nonumber \\
 & -d\left(k^{2}(1+2n-\eta)+(1+s+\eta)(2s+4ns+\eta+n\eta)+k(-1-3\ s+\eta(s+\eta)-3n(1+2s+\eta))\right)\big)\big\}.\nonumber 
\end{align}
Hence, we solve the first order recurrence defined by $Qt_{n,\eta,k,s}$
to obtain 
\begin{align}
t_{n,\eta,k,s} & =-\frac{(\eta+s-k)}{(n-\eta+k-s+1)}t_{n,\eta,k,s-1}\nonumber \\
 & =(-1)^{s}\frac{(\eta+s-k)!}{(\eta-k)!}\frac{(n-\eta+k-s)!}{(n-\eta+k)!}t_{n,\eta,k,0}\nonumber \\
 & =(-1)^{s}\frac{(\eta+s-k)!}{(\eta-k)!}\frac{(n-\eta+k-s)!}{(n-\eta+k)!}\frac{(n-\eta+k)!}{k!(n-\eta)!}\nonumber \\
 & =(-1)^{s}\binom{k}{s}^{-1}\binom{\eta+s-k}{s}\binom{n-\eta+k-s}{k-s}.
\end{align}
\end{proof}

\end{document}